\documentclass{lmcs} 
\pdfoutput=1

\usepackage{lastpage}
\lmcsdoi{16}{1}{1}
\lmcsheading{}{\pageref{LastPage}}{}{}%
{Sep.~12,~2018}{Jan.~06,~2020}{}

\keywords{Cyclic proofs, Proof theory, Logical complexity, Peano arithmetic, Induction}

\usepackage{hyperref}
\usepackage{amsmath}
\usepackage{amsthm}
\usepackage{amssymb}
\usepackage{virginialake}
\usepackage{enumitem}
\usepackage{pdflscape} 
\usepackage{microtype}

\renewcommand{\phi}{\varphi}
\renewcommand{\emptyset}{\varnothing}
\newcommand{\dfn}{:=}
\newcommand{\dfntrm}[1]{\textbf{{#1}}}

\newcommand{\proves}{\vdash}
\newcommand{\notproves}{\nvdash}
\renewcommand{\models}{\vDash}
\newcommand{\notmodels}{\nvDash}

\renewcommand{\vlvruler}{\hbox{\vrule width .6pt height \vlstemheightc sp
	}}
	
\newcommand{\Nat}{\mathbb{N}}

\newcommand{\lang}{\mathcal{L}}

\newcommand{\id}{\mathit{id}}
\newcommand{\sub}{\mathit{sub}}
\newcommand{\cut}{\mathit{cut}}
\newcommand{\wk}{\mathit{wk}}

\newcommand{\rul}{\mathsf{r}}
\newcommand{\lefrul}[1]{#1\text{-}l}
\newcommand{\rigrul}[1]{#1\text{-}r}

\newcommand{\axeqrefn}{=_1}
\newcommand{\axeqfn}{=_2}
\newcommand{\axeqpred}{=_3}

\newcommand{\daemon}{\bullet}

\newcommand{\cand}{\wedge}
\newcommand{\corr}{\vee}
\newcommand{\cimp}{\supset}
\newcommand{\cnot}{\neg}
\newcommand{\ciff}{\equiv}
\newcommand{\seqar}{\Rightarrow}

\newcommand{\prf}{\mathsf{Prf}}
\newcommand{\Prov}{\square_n }
\newcommand{\cons}[1]{\mathsf{Con}({#1})}
\newcommand{\consd}[1]{\mathsf{Con}_d({#1})}

\newcommand{\ind}{\mathit{ind}}

\newcommand{\Branch}{\mathrm{Branch}}
\renewcommand{\succ}{\mathsf{s}}
\newcommand{\Rfn}[2]{#1\text{-}\mathsf{Rfn}({#2})}
\newcommand{\Rfnd}[2]{#1\text{-}\mathsf{Rfn}_d({#2})}

\newcommand{\ArAcc}{\mathrm{ArAcc}}

\newcommand{\Sin}[2]{{\Sigma}^{#1}_{#2}}
\newcommand{\Pin}[2]{{\Pi}^{#1}_{#2}}
\newcommand{\Din}[2]{{\Delta}^{#1}_{#2}}

\newcommand{\IC}[1]{\mathrm{I}#1}
\newcommand{\CC}[1]{\mathrm{C}#1}

\newcommand{\ISn}[1]{\IC{\Sin{}{#1}}}
\newcommand{\IPn}[1]{\IC{\Pin{}{#1}}}
\newcommand{\IDn}[1]{\IC{\Din{}{#1}}}
\newcommand{\CSn}[1]{\CC{\Sin{}{#1}}}
\newcommand{\CPn}[1]{\CC{\Pin{}{#1}}}
\newcommand{\CDn}[1]{\CC{\Din{}{#1}}}

\newcommand{\bool}{\mathrm{Bool}}

\newcommand{\CIND}[1]{#1\text{-}\mathsf{IND}}
\newcommand{\CCA}[1]{#1\text{-}\mathsf{CA}}

\newcommand{\pspace}{\mathbf{PSPACE}}
\newcommand{\lth}{\mathbf{LTH}}
\newcommand{\flth}{\mathbf{FLTH}}
\newcommand{\np}{\mathbf{NP}}

\newcommand{\PA}{\mathsf{PA}}

\newcommand{\RCA}{\mathsf{RCA}_0}
\newcommand{\Q}{\mathsf{Q}}
\newcommand{\CA}{\mathsf{CA}}

\newcommand{\McN}{\mathsf{McNaughton}^-}
\newcommand{\wkl}{\mathsf{WKL}}
\newcommand{\bwwkl}{\mathsf{bwWKL}}
\newcommand{\LKID}{\mathsf{LKID}}
\newcommand{\CLKID}{\mathsf{CLKID}}
\newcommand{\FOLID}{\mathrm{FOL}_{\mathit{ID}}}

\newcommand{\PHP}{\mathsf{PHP}}
\newcommand{\Inj}{\mathrm{Inj}_f }
\newcommand{\inn}{\in}

\newcommand{\lift}[1]{\lceil #1 \rceil}

\newcommand{\compl}[1]{\overline{#1}}

\newcommand{\univ}{\mathrm{Univ}}
\newcommand{\Empty}{\mathrm{Empty}}

\newcommand{\conc}{\mathsf{conc}}

\newcommand{\pair}[2]{\langle #1,#2\rangle}

\newcommand{\trarr}[2]{\underset{#2}{\overset{#1}{\rightarrow}}}
\newcommand{\trrra}[2]{\underset{#2}{\overset{#1}{\leftarrow}}}
\newcommand{\longtrarr}[2]{\underset{#2}{\overset{#1}{\longrightarrow}}}

\newcommand{\final}{\infty}

\theoremstyle{plain} 

\begin{document}

\title{On the logical complexity of cyclic arithmetic}

\author[A.~Das]{Anupam Das}	
\address{University of Birmingham, United Kindgdom}	
\email{a.das@bham.ac.uk}  
\thanks{The author is supported by a Marie Sk\l{}odowska-Curie fellowship, \href{http://cordis.europa.eu/project/rcn/209401_en.html}{ERC project 753431}.}	





\begin{abstract}
  \noindent 
  We study the logical complexity of proofs in cyclic arithmetic ($\CA$), as introduced by Simpson in \cite{Sim17:cyclic-arith}, in terms of quantifier alternations of formulae occurring. Writing $\CSn n$ for (the logical consequences of) cyclic proofs containing only $\Sin{}{n}$ formulae, our main result is that $\ISn{n+1}$ and $\CSn n$ prove the same $\Pin{}{n+1}$ theorems, for $n \geq 0$. Furthermore, due to the `uniformity' of our method, we also show that $\CA$ and Peano Arithmetic ($\PA$) proofs of the same theorem differ only exponentially in size. 
  
  The inclusion $\ISn{n+1} \subseteq \CSn n$ is obtained by proof theoretic techniques, relying on normal forms and structural manipulations of $\PA$ proofs. It improves upon the natural result that $\ISn n \subseteq \CSn n$.
  The converse inclusion, $\CSn n \subseteq \ISn{n+1}$, is obtained by calibrating the approach of \cite{Sim17:cyclic-arith} with recent results on the reverse mathematics of B\"uchi's theorem \cite{KMPS16:buchi-reverse}, and carefully specialising to the case of cyclic proofs. These results improve upon the bounds on proof complexity and logical complexity implicit in \cite{Sim17:cyclic-arith} and \cite{BerTat17:lics}. 
  
  The uniformity of our method also allows us to recover a metamathematical account of fragments of $\CA$; in particular we show that, for $n\geq 0$, the {consistency} of $\CSn{n}$ is provable in $\ISn{n+2}$ but not $\ISn{n+1} $.
  As a result, we show that certain versions of McNaughton's theorem (the determinisation of $\omega$-word automata) are not provable in $\RCA$, partially resolving an open problem from \cite{KMPS16:buchi-reverse}.
\end{abstract}

\maketitle

\section{Introduction}

\emph{Cyclic} and \emph{non-wellfounded} proofs have been studied by a number of authors as an alternative to proofs by induction.
This includes cyclic systems for fragments of the modal $\mu$-calculus, e.g.\ \cite{NiwWal96:games-mucalc,SprDam03:mucalc,DaxHofLan06:multl,BaeDouHirSau16:multl,Dou17:multl,AfsLei17:cut-free-mucalc}, structural proof theory for logics with fixed-points, e.g.\ \cite{San02:circ-proofs-categorical-semantics,ForSan13:cuts-circ-proofs-sem-cut-elim,For14:phdthesis,BaeDouSau16:cut-elim}, (automated) proofs of program termination in separation logic, e.g.\ \cite{BroBorCal08:cyc-proofs-term,BroDisPet11:aut-cyc-ent,BroRow17:aut-cyc-term} and, in particular, cyclic systems for first-order logic with inductive definitions, e.g.\ \cite{Bro05:cyc-proofs-folid,Bro06:phdthesis,BroSim07:comp-seq-calc-ind-inf-desc,BroSim11:seq-calc-ind-inf-desc}.
Due to the somewhat implicit nature of invariants they define, cyclic systems can be advantageous for metalogical analysis, for instance offering better algorithms for proof search,~e.g.~\cite{BroGorPet12:gen-cyc-prover,DasPou17:cut-free-cyc-prf-sys-kl-alg}.

Cyclic proofs may be seen as more intuitively analogous to proofs by `infinite descent' than proofs by induction (see, e.g., \cite{Sim17:cyclic-arith}); this subtle difference is enough to make inductive invariants rather hard to generate from cyclic proofs. 
Indeed it was recently shown that simulating cyclic proofs using induction is not possible for some sub-arithmetic languages \cite{BerTat17:cfolid-neq-folid}, but becomes possible once arithmetic reasoning is available \cite{Sim17:cyclic-arith,BerTat17:lics}.


\emph{Cyclic arithmetic} was proposed as a general subject of study by Simpson in \cite{Sim17:cyclic-arith}.
Working in the language of arithmetic, it replaces induction by non-wellfounded proofs with a certain `fairness' condition on the infinite branches.
The advantage of this approach to infinite proof theory as opposed to, say, infinite well-founded proofs via an $\omega$-rule (see, e.g., \cite{Sch77:prf-th}), is that it admits a notion of \emph{finite proofs}: those that have only finitely many distinct subproofs, and so may be represented by a finite (possibly cyclic) graph.

Cyclic arithmetic itself is to cyclic proofs what Peano arithmetic is to traditional proofs: it provides a general framework in which many arguments can be interpreted and/or proved in a uniform manner, and this is one reason why it is an interesting subject of study.
This is already clear from, say, the results of \cite{BerTat17:lics}, where the study of cyclic proofs for first-order logic with inductive definitions relied on an underlying arithmetic framework.
We elaborate further on this in Sect.~\ref{sect:conc}.


\subsection*{Contribution}
In \cite{Sim17:cyclic-arith}, Simpson showed that Peano Arithmetic ($\PA$) is able to simulate cyclic reasoning by proving the {soundness} of the latter in the former. (The converse result is obtained much more easily.)
Nonetheless, several open questions remain from \cite{Sim17:cyclic-arith}, concerning constructivity, normalisation, logical complexity and proof complexity for cyclic and non-wellfounded proofs.

In this work we address the \emph{logical complexity} and \emph{proof complexity} of proofs in Cyclic Arithmetic ($\CA$), as compared to $\PA$.
Namely, we study how quantifier alternation of proofs in one system compares to that in the other, and furthermore how the size of proofs compare.
Writing $\CSn{n}$ for (the logical consequences of) cyclic proofs containing only $\Sin{}{n}$ formulae, we show, for $n \geq 0$:

\begin{enumerate}[start=1,label={(\arabic*)}]
	\item\label{item:cyc-sim-ind} $\ISn{n+1} \subseteq \CSn n$ over $\Pin{}{n+1}$ theorems (Sect.~\ref{sect:cyc-sim-ind}, Thm.~\ref{thm:cyclic-sim-ind}).
	\item\label{item:proof-complexity} $\CA$ and $\PA$ proofs of the same theorem differ only exponentially in size (Sect.~\ref{sect:ind-sim-cyc}, Thm.~\ref{thm:elementary-simulation}).
	\item\label{item:ind-sim-cyc} $\CSn n \subseteq \ISn{n+1}$ over all theorems (Sect.~\ref{sect:nonuniform-ind-sim-cyc}, Thm.~\ref{thm:nonuniform-ind-sim-cyc}).
\end{enumerate}

\ref{item:cyc-sim-ind} is obtained by proof theoretic techniques, relying on normal forms and structural manipulations of Peano Arithmetic proofs. It improves upon the natural result that $\ISn n \subseteq \CSn n$, although induces a non-elementary blowup in the size of proofs.
\ref{item:proof-complexity} is obtained via a certain `uniformisation' of the approach of \cite{Sim17:cyclic-arith}. In particular, by specialising the key intermediate results to the case of cyclic proofs, we are able to extract small $\PA$ proofs of some required properties of infinite word automata from analogous ones in `second-order' (SO) arithmetic.
Finally, \ref{item:ind-sim-cyc} is obtained by calibrating the argument of \ref{item:proof-complexity}
with recent results on the reverse mathematics of B\"uchi's theorem \cite{KMPS16:buchi-reverse},
allowing us to bound the logical complexity of proofs in the simulation.
Together, these results almost completely characterise the logical and proof complexity theoretic strength of cyclic proofs in arithmetic, answering the questions (ii) and (iii), Sect.~7 of \cite{Sim17:cyclic-arith}.


After demonstrating these results, we give a metamathematical analysis of provability in cyclic theories, in particular showing that the consistency of $\CSn n $ is provable in $\ISn{n+2}$ but not $\ISn{n+1}$, by appealing to a form of \emph{G\"odel incompleteness} for cyclic theories.
We also use these observations to show that certain formulations of McNaughton's theorem, that every nondeterministic B\"uchi automaton has an equivalent deterministic parity (or Rabin, Muller, etc.) automaton, are not provable in the SO theory $\RCA$. 
This partially resolves the question of the logical strength of McNaughton's theorem left open in \cite{KMPS16:buchi-reverse}.

\subsection*{Structure of the paper}
In Sects.~\ref{sect:prelims-pa-pt} and \ref{sect:prelims-ca-aut} we introduce some preliminaries on Peano Arithmetic, proof theory, cyclic proofs and automaton theory.
In Sect.~\ref{sect:cyc-sim-ind} we present \ref{item:cyc-sim-ind} and give an example of the translation in App.~\ref{sect:php-case-study}. The contents of Sects.~\ref{sect:prelims-pa-pt}, \ref{sect:prelims-ca-aut} and \ref{sect:cyc-sim-ind} and App.~\ref{sect:php-case-study} are more-or-less self contained and should be accessible to the general proof theorist.

We briefly introduce some SO theories of arithmetic in Sect.~\ref{sect:so-theories} that are conservative over the fragments of $\PA$ we need in order to conduct some of the intermediate arguments on infinite word automata.
In Sect.~\ref{sect:ind-sim-cyc} we present \ref{item:proof-complexity}, and in Sect.~\ref{sect:nonuniform-ind-sim-cyc} we adapt the argument to obtain \ref{item:ind-sim-cyc}.
In Sect.~\ref{sect:further-metalogical} we give our metamathematical analysis of cyclic theories, and in Sect.~\ref{sect:red-to-det} we explain their consequences for the logical strength of certain forms of McNaughton's theorem.
We conclude with some further remarks and perspectives in Sect.~\ref{sect:conc}, including a comparison with the results of \cite{Sim17:cyclic-arith} and \cite{BerTat17:lics}.

For Sects.~\ref{sect:so-theories}, \ref{sect:ind-sim-cyc}, \ref{sect:nonuniform-ind-sim-cyc} and \ref{sect:red-to-det} it would be helpful for the reader to have some background in subsystems of second-order arithmetic (see e.g.~\cite{Sim09:reverse-math,Hir14:reverse-math}) and $\omega$-automaton theory (see e.g.~\cite{Tho97:aut-chapt}).
For Sect.~\ref{sect:further-metalogical} it would be helpful for the reader to have some background in the metamathematics of first-order arithmetic (see e.g.~\cite{HajPud:93}). 
Nonetheless we aim to give sufficient details for the general proof theorist to appreciate all the content herein.

%
%
%



\section{Preliminaries on first-order arithmetic proof theory}
\label{sect:prelims-pa-pt}
We present only brief preliminaries, but the reader is encouraged to consult, e.g., \cite{Bus98:handbook-of-pt} for a more thorough introduction to first-order arithmetic.

%
We work in first-order (FO) logic with equality, $=$, with variables written $x,y,z $ etc., terms written $s,t,u$ etc., and formulae written $\phi, \psi$ etc., construed over the logical basis $\{ \cnot, \corr ,\cand , \exists , \forall \}$.
We will usually assume formulae are in \textbf{De Morgan normal form}, with negation restricted to atomic formulae.
Nonetheless, we may write $\cnot \phi$ for the De Morgan `dual' of $\phi$, defined as follows:
\[
\neg \neg \phi \dfn \phi
\quad
\begin{array}{rcl}
\neg (\phi \cand \psi) &\dfn &\neg \phi \corr \neg \psi 
\\
\neg (\phi \corr \psi) &\dfn &\neg \phi \cand \neg \psi
\end{array} 
\quad
\begin{array}{rcl}
\neg \forall x . \phi &\dfn &\exists x . \neg \phi
\\
\neg \exists x . \phi &\dfn &\forall x . \neg \phi
\end{array}
\]
We also write $\phi \cimp \psi $ for $\cnot \phi \corr \psi$ and $\phi \ciff \psi $ for $(\phi \cimp \psi ) \cand (\psi \cimp \phi)$.

FO logic has equality `built-in', i.e.\ we always assume the following axioms are present:
\begin{enumerate}[start=1,label={(eq\arabic*)}]
	\item\label{item:eq-ax-refl} $\forall x . x = x $. 
	\item\label{item:eq-ax-fn} $\forall \vec x , \vec y . ( (x_1 = y_1 \cand \cdots \cand x_k = y_k) \cimp f(\vec x ) = f(\vec y)$, for each $k\in \Nat$ and each function symbol $f$ of arity $k$.
	\item\label{item:eq-ax-pred} $\forall x,y . (((x_1 = y_1 \cand \cdots \cand x_k = y_k)\cand P(\vec x)) \cimp P(\vec y))$, for each $k \in \Nat$ and each predicate symbol $P$ of arity $k$.
\end{enumerate}

\noindent
Following \cite{Sim17:cyclic-arith}, the \textbf{language of arithmetic} (with inequality) is formulated as $\{ 0, \succ, +, \times , < \}$, with their usual interpretations over $\Nat$.
A \textbf{theory} is a set $T$ of closed formulae over this language.
We write $T\proves \phi$ if $\phi$ is a logical consequence of $T$.
We write $T_1 \subseteq T_2 $ if $T_1 \proves \phi$ implies $T_2 \proves \phi$, and $T_1 = T_2$ if $T_1 \subseteq T_2$ and $T_2 \subseteq T_1$.

The theory of \textbf{Robinson arithmetic} (with inequality), written $\Q$, is axiomatised by:

\begin{enumerate}[start=1,label={(Q\arabic*)}]
	\item $\forall x . \succ x \neq 0 $. 
	\item $\forall x, y. (\succ x = \succ y \cimp x = y)$.
	\item\label{item:rob-succ-non-zer} $\forall x . (x \neq 0 \cimp \exists y . x = \succ y)$.
	\item $\forall x .\  x + 0 = x$.
	\item $\forall x , y .\  x + \succ y = \succ (x + y)$.
	\item $\forall x .\  x \cdot 0 = 0 $.
	\item $\forall x , y .\  x \cdot \succ y = x\cdot y + x$.
	\medskip
	\item\label{item:rob-ineq-dfn} $\forall x , y . (x<y \ciff \exists z . (x + \succ z = y) )$
\end{enumerate}

%
\noindent
Notice that, above and elsewhere, we may write $\cdot$ instead of $\times $ in terms, or even omit the symbol altogether, and we assume it binds more strongly than $+$.
%
%
We also write $\forall x< t. \phi$ and $\exists x < t. \phi$ as abbreviations for $\forall x . (x< t \cimp \phi)$ and $\exists x. ( x < t \cand \phi )$ resp.
Formulae with only such quantifiers are called \textbf{bounded}.

As usual, we may assume that $\Q$ is axiomatised by the universal closures of bounded formulae. In particular the existential quantifiers in axioms \ref{item:rob-succ-non-zer} and \ref{item:rob-ineq-dfn} above may be bounded by $x$ and $y$ resp., provably under quantifier-free induction.
We will implicitly assume this bounded axiomatisation for the sequent calculus formulation of arithmetic later.

\begin{rem}
	\label{rmk:diff-ax-simpson}
	Our basic axioms and, later, our inference rules differ slightly from those in \cite{Sim17:cyclic-arith}, however it is routine to see that the theories $\PA$ and $\CA$ defined in this work coincide with those of \cite{Sim17:cyclic-arith}.
	In particular the axiomatisations are equivalent once even open induction is present (this is weaker than any theory we will consider).
	We chose a slightly different presentation so that we could readily apply certain metalogical results, such as Thm.~\ref{thm:free-cut-elim}, with no intermediate proof manipulation.
\end{rem}

\begin{defi}
	[Arithmetical hierarchy]
	For $n\geq 0$, we define:
	\begin{itemize}
		\item $\Din{}{0} = \Pin{}{0} = \Sin{}{0}$ is the class of bounded formulae.
		\item $\Sin{}{n+1}$ is the class of formulae of the form $\exists \vec x . \phi$, where $\phi \in \Pin{}{n}$.
		\item $\Pin{}{n+1}$ is the class of formulae of the form $\forall \vec x . \phi$, where $\phi \in \Sin{}{n}$.
	\end{itemize}
\end{defi}
%
%
\noindent
Notice in particular that, by definition of De Morgan normal form, if $\phi  \in \Sin{}{n}$ then $\cnot \phi \in \Pin{}{n}$ and vice-versa.
{In practice we often consider these classes of formulae up to logical equivalence.}
We say that a formula is in $\Din{}{n}$ (in a theory $T$) if it is equivalent (resp.\ provably equivalent in $T$) to both a $\Sin{}{n}$ and $\Pin{}{n}$ formula.

\begin{defi}
	[Arithmetic]
	\textbf{Peano Arithmetic} ($\PA$) is axiomatised by $\Q$ and the axiom schema of \textbf{induction}:
	\begin{equation}
	\label{eqn:ind-ax}
	( \phi(0) \cand \forall x . (\phi(x) \cimp \phi(\succ x)) ) \cimp \forall x . \phi(x)
	\end{equation}
	For a class of formulae $\Phi$, we write $\CIND{\Phi}$ for the set of induction axiom instances when $\phi \in \Phi$ in \eqref{eqn:ind-ax}.
	We write $\IC{\Phi}$ for the theory $\Q + \CIND{\Phi}$.

\end{defi}

The following is a classical result:
\begin{prop}
	[See e.g.~\cite{Bus98:handbook-of-pt,Kay91:models-of-pa}]
	\label{prop:ipn-equals-isn}
	For $n\geq 0$, we have $\ISn{n} = \IPn{n}$.
\end{prop}


\subsection{A sequent calculus presentation of $\PA$}

\begin{figure}
	\[
	\begin{array}{ccc}
	\vlinf{
		\id
	}{}{\Gamma, \phi \seqar \Delta, \phi}{} 
	\ & \ 
	\vlinf{\axeqrefn}{}{\Gamma \seqar \Delta , t=t}{}
	\ & \ 
	\vlinf{\axeqfn}{}{\Gamma, s_1 = t_1 , \dots , s_k = t_k \seqar \Delta, f(\vec s) = f(\vec t)}{}
	\\
	\vlinf{\lefrul\cnot}{}{\Gamma , \phi, \cnot \phi \seqar \Delta}{}
	\ & \ 
	\vlinf{\rigrul \cnot}{}{\Gamma \seqar \Delta, \phi, \cnot \phi}{} 
	\ & \ 
	\vlinf{\axeqpred}{}{\Gamma, s_1 = t_1 , \dots , s_k = t_k , P(\vec s)\seqar \Delta, P(\vec t)}{}
	\end{array}
	\]
	\medskip
	\[
	\begin{array}{ccc}
	\vlinf{
		\theta\text{-}
		\sub
	}{}{\theta(\Gamma) \seqar \theta(\Delta)}{\Gamma \seqar \Delta}
	\ & \ 
	\vliinf{
		\cut
	}{}{\Gamma \seqar \Delta}{\Gamma \seqar \Delta , \phi}{\Gamma, \phi \seqar \Delta}
	\ & \ 
	\vlinf{
		\wk
	}{}{\Gamma, \Gamma'\seqar \Delta , \Delta'}{\Gamma \seqar \Delta}
	\end{array}
	\]
	\medskip
	\[
	\hspace{-1em}
	\begin{array}{cccc}
	\vliinf{\lefrul \corr}{}{\Gamma, \phi \corr \psi \seqar \Delta }{ \Gamma , \phi \seqar \Delta }{\Gamma, \psi \seqar \Delta}
	&
	\vlinf{\lefrul \cand}{
	}{\Gamma , \phi_0 \cand \phi_1 \seqar \Delta}{\Gamma, \phi_i \seqar \Delta }
	&
	\vlinf{\lefrul \exists}{
	}{\Gamma , \exists x . \phi \seqar \Delta}{\Gamma , \phi [a/x] \seqar \Delta}
	& \vlinf{\lefrul \forall}{}{\Gamma , \forall x . \phi \seqar \Delta}{\Gamma, \phi [t/x] \seqar \Delta}
	\\
	\noalign{\medskip}
	\vliinf{\rigrul \cand}{}{\Gamma \seqar \Delta, \phi \cand \psi}{\Gamma \seqar \Delta, \phi }{\Gamma \seqar \Delta, \psi }
	&
	\vlinf{\rigrul \corr}{
	}{\Gamma\seqar \Delta , \phi_0 \corr \phi_1}{ \Gamma \seqar \Delta , \phi_i}
	&
	\vlinf{\rigrul \forall}{
	}{\Gamma \seqar \Delta, \forall x . \phi}{\Gamma  \seqar \Delta, \phi [a/x] }
	& \vlinf{\rigrul \exists}{}{\Gamma \seqar \Delta, \exists x. \phi}{\Gamma\seqar \Delta, \phi [t/x] }
	\end{array}
	\]
	\caption{The sequent calculus for FO logic with equality, where $a$ occurs only as indicated and $i \in \{0,1\}$.}
	\label{fig:seq-calc}
\end{figure}
We will work with a standard sequent calculus presentation of FO logic, given in Fig.~\ref{fig:seq-calc}, where $i \in \{0,1\}$ and $a$, known as the `eigenvariable', is fresh, i.e.\ does not occur free in the lower sequent.
Two important considerations are that we work with cedents as \emph{sets}, i.e.\ there is no explicit need for contraction rules, and that we have an explicit \textbf{substitution} rule.
{In the $\theta$-$\mathit{sub}$ rule the `substitution' $\theta$ is a mapping from variables to terms, which is extended in the natural way to cedents.}
Substitution is important for the definition of a cyclic arithmetic proof in the next section, but does not change provability in usual proofs. 
%

The sequent calculus for $\Q$ is obtained from the FO calculus in the language of arithmetic by adding appropriate initial sequents for each instantiation of an axiom of $\Q$ by terms.
For theories extending $\Q$ by (at least quantifier-free) induction, we assume that these intial sequents contain only $\Din{}{0}$ formulae by appropriately bounding the existential quantifiers. The schema $\CIND{\Phi}$, for $\Phi$ closed under subformulas and substitution, is implemented in the calculus by adding the induction rule,
\[
\vliinf{\ind}{}{\Gamma \seqar \phi(t), \Delta}{\Gamma \seqar \phi (0), \Delta}{\Gamma , \phi(a) \seqar \phi(\succ a), \Delta }
\]
for formulae $\phi \in \Phi$.
{Here we require $a$ to not occur free in the lower sequent.} {Notice that this satisfies the subformula property, in the `wide' sense of FO logic, i.e.\ up to substitution.}
For fragments of $\PA$ with induction axioms of bounded logical complexity, we also have the \emph{bounded quantifier rules}:
\[
\begin{array}{cc}
\vlinf{}{}{\Gamma,\exists x < s . \phi(x) \seqar \Delta}{\Gamma, a < s , \phi(a) \seqar \Delta}
\quad &\quad 
\vlinf{}{}{\Gamma, t< s \seqar \Delta, \exists x < s . \phi(x)}{\Gamma \seqar \Delta , \phi(t)}
\\
\noalign{\bigskip}
\vlinf{}{}{ \Gamma \seqar \Delta , \forall x < s . \phi(x) }{ \Gamma,a< s \seqar \Delta , \phi(a) }
\quad &
\quad
\vlinf{}{}{\Gamma , t< s , \forall x < s . \phi(x) \seqar \Delta}{\Gamma , \phi(t) \seqar \Delta}
\end{array}
\]
In all cases the eigenvariable $a$ occurs only as indicated,

The following normalisation result is well-known in the proof theory of arithmetic, and will be one of the main structural proof theoretic tools in this work:

\begin{thm}
	[Free-cut elimination, e.g.\ \cite{Bus98:handbook-of-pt}]
	\label{thm:free-cut-elim}
	Let $\mathcal S$ be a sequent system extending FO by the induction rule and some other nonlogical rules/axioms closed under substitution.
	Then any $\mathcal S$-proof can be effectively transformed into one of the same conclusion containing only (substitution instances of) subformulae of the conclusion, an induction formula or a formula occurring in another nonlogical step.
\end{thm}
Naturally, this applies to the various fragments of $\PA$ that we consider. In particular, notice that a free-cut free proof in $\ISn{n}$ or $\IPn{n}$ of sequents containing only $\Sin{}{n}$ or $\Pin{}n$ formulae, resp., contains just $\Sin{}{n}$ or $\Pin{}{n}$ formulae, resp.
It is well known that Thm.~\ref{thm:free-cut-elim} can itself be proved within $\ISn 1$ and even weaker theories (see, e.g., \cite{HajPud:93}), under an appropriate coding of mathematical objects. We use this observation later in Sect.~\ref{sect:further-metalogical}.

\medskip

We say that a sequent is $\Sin {}  n$ (or $\Pin{}{n}$) if it contains only $\Sin{}{n}$ (resp.\ $\Pin {} n$) formulae.
A slight issue that will be relevant later in Sect.~\ref{sect:cyc-sim-ind} is that we have not defined $\Sin{}{n}$ and $\Pin{}{n}$ as being syntactically closed under positive Boolean combinations, even if semantically we know that they are.
In fact, this does not cause a problem for the result above, since we can always prenex `on the fly' in a proof by cutting against appropriate derivations.
For instance, in a proof, a step of the form,
\[
\vliinf{\cand}{}{ \Gamma \seqar \Delta , \forall x. \phi \cand \forall y . \psi }{\Gamma \seqar \Delta, \forall x . \phi}{\Gamma \seqar \Delta, \forall y . \psi}
\]
may be locally replaced by a derivation of the form:
\[\vlderivation{
	\vliin{\cut}{}{ \Gamma \seqar \Delta, \forall x , y. (\phi \cand \psi) }{
		\vlhy{\Gamma \seqar \Delta, \forall x . \phi}
	}{
		\vliin{\cut}{}{ \Gamma , \forall y . \psi \seqar \Delta, \forall x,y . (\phi \cand \psi) }{
			\vlhy{\Gamma \seqar \Delta , \forall y . \psi}
		}{
			\vliq{}{}{\forall x . \phi , \forall y . \psi \seqar \forall x , y . (\phi \cand \psi) }{\vlhy{}}
		}
	}
}
\]
In a similar way we will often assume that a `block' of existential or universal quantifiers is coded by a single quantifier, using pairings and G\"odel $\beta$ functions, whose basic properties are all formalisable already in $\IDn 0$ (see, e.g., \cite{Bus98:handbook-of-pt}).

\section{Preliminaries on cyclic arithmetic and automata}
\label{sect:prelims-ca-aut}

Before presenting `cyclic arithmetic', we will present the general notion of non-wellfounded proofs in arithmetic, from \cite{Sim17:cyclic-arith}.

By convention, we say \emph{binary tree} to mean a nonempty set $T\subseteq \{0,1\}^*$ that is prefix-closed, i.e.\ if $\sigma i \in T$ then $\sigma \in T$.
We construe such $T$ as a bona fide tree with nodes $T$ and directed edges from $\sigma$ to $\sigma i $, if $\sigma i \in T$, for $i \in \{0,1\}$.
The empty word, $\varepsilon$, is the \emph{root} of $T$.
\begin{defi}
	A \textbf{preproof} is a possibly infinite binary tree labelled by sequents and rules in a locally correct manner in the calculus for $\Q$. {Following \cite{Sim17:cyclic-arith}, we treat inference steps as nodes of the tree and sequents as edges.}
	A preproof is \textbf{regular} if it has only finitely many distinct (labelled) subtrees or, equivalently, if it is the unfolding of a finite labelled directed graph, possibly with cycles.
\end{defi}

The following notions are variants of those from Dfns.~1 and 2 in \cite{Sim17:cyclic-arith}:
\begin{defi}
	[Precursors, traces, $\infty$-proofs]
	Let $(\Gamma_i \seqar \Delta_i )_{i \geq 0} $ be an infinite branch through a preproof.
	For terms $t, t'$ we say that $t'$ is a \dfntrm{precursor} of $t$ at $i$ if one of the following holds:
	\begin{enumerate}[label={\roman*)}]
		\item\label{item:prec-sub} $\Gamma_i \seqar \Delta_i$ concludes a $\theta$-$\sub$-step and $t $ is $ \theta (t')$.
		\item\label{item:prec-eq} $\Gamma_i \seqar \Delta_i$ concludes any other step and $t' = t$ occurs in $ \Gamma_i$.
		\item\label{item:prec-id} $\Gamma_i \seqar \Delta_i$ concludes any other step and $t'$ is $t$.
	\end{enumerate}
	\noindent
	A \dfntrm{trace} along $(\Gamma_i \seqar \Delta_i)_{i \geq 0}$ is a sequence $(t_i)_{i \geq n}$, for some $n\geq 0$, such that whenever $i\geq n$ the term $t_i$ occurs in $\Gamma_i \seqar \Delta_i $ and,

	\begin{enumerate}[label={(\alph*)}]
		\item 
		$t_{i+1}$ is a precursor of $t_i$ at $i$; or
		\item\label{item:progress-point} 
		the atomic formula $t_{i+1} < t $ occurs in $\Gamma_{i+1}$, where $t$ is a precursor of $t_i$ at $i$.
	\end{enumerate}
	When 
	\ref{item:progress-point} 
	holds, we say that the trace \dfntrm{progresses} at $i+1$.
	
	An \dfntrm{$\infty$-proof} is a preproof for which any infinite branch has a trace that progresses infinitely often.
	If it is regular then we simply call it a \dfntrm{cyclic proof}.
	$\CA$ is the theory induced by cyclic proofs in the calculus for $\Q$.
\end{defi}


\begin{rem}
	\label{rmk:traces-occurrences}
	When defining explicit traces, for the precursor case \ref{item:prec-id}, we will typically not worry about whether the term $t_i $ in a trace occurs in the sequent or not. All that matters is that, if the current step is $\lefrul \exists$, $\rigrul \forall$ or $\sub$, $t_{i}$ does not contain the associated eigenvariables.\footnote{We say that $a$ is an eigenvariable of a $\theta$-$\sub$ step if it is in the support of the substitution $\theta$.} As long as we satisfy this constraint we may simply consider an equivalent proof that prepends $t_i=t_i$ to the antecedent to make sure that $t_i$ `occurs'. We use this assumption implicitly in the remainder of this work. 
\end{rem}

The reader may consult \cite{Sim17:cyclic-arith} for several examples of $\infty$-proofs.
Notably, $\infty$-proofs are sound and {complete} for the standard model $\Nat$ (Thm.~4, \cite{Sim17:cyclic-arith}).
%
%
%
%
(Similar results for other logics, with respect to standard models, were known before \cite{Bro06:phdthesis,BroSim11:seq-calc-ind-inf-desc}.)
We recall the proof of soundness since we will have to formalise a variant of it in Sect.~\ref{sect:ind-sim-cyc}, and also since the quantifier case in the argument of \cite{Sim17:cyclic-arith} is omitted, whereas this subtlety will need some consideration when it is formalised.

\begin{prop}
	[Soundness of $\infty$-proofs]
	\label{prop:sound-cyclic}
	If $\pi$ is an $\infty$-proof of $\phi$, then $\Nat \models \phi$.
\end{prop}
\begin{proof}
	Suppose otherwise, i.e.\ $\Nat \models \cnot \phi$.
	We will inductively construct an infinite branch $(\Gamma_i \seqar \Delta_i)_{i \geq 0}$ of $\pi$ and associated assignments $\rho_i$ of natural numbers to each sequent's free variables, such that $\Nat , \rho_i \notmodels \Gamma_i \seqar \Delta_i$.
	Assuming $\phi$ is closed (by taking its universal closure), we set $\Gamma_0 \seqar \Delta_0 $ to be $\seqar \phi$ and $\rho_0 = \emptyset$.
	
	Each step except for substitution, $\rigrul{\forall}$ and $\lefrul{\exists}$ constitutes a true implication, so if $\Nat, \rho_i \notmodels \Gamma_i \seqar \Delta_i$ then $\rho_i$ also must not satisfy one of its premisses.
	We may thus choose one such premiss as $\Gamma_{i+1} \seqar \Delta_{i+1}$ and set $\rho_{i+1} = \rho_i$.

	If $\Gamma_i \seqar \Delta_i$ concludes a $\theta$-$\sub$ step, we may set $\rho_{i+1} = \rho_i \circ \theta$.
	If $\Gamma_i \seqar \Delta_i$ concludes a $\rigrul{\forall}$ step, let $\forall x . \phi $ be the principal formula and assume $x $ does not occur free in the conclusion.
	Since $\Nat, \rho_i \notmodels \Gamma_i \seqar \Delta_i$, we must have that $\Nat , \rho_i \models \exists x . \cnot \phi$. We choose a value $k \in \Nat$ witnessing this existential and set $\rho_{i+1} = \rho_i \cup \{ x \mapsto k \}$.
	The $\lefrul{\exists}$ case is dealt with similarly.
	
	This infinite branch must have an infinitely progressing trace, say $(t_i)_{i \geq n}$, by the definition of $\infty$-proof.
	However notice that, for $i \geq n$, $\rho_i (t_i) \geq \rho_{i+1} (t_{i+1})$ and, furthermore, at a progress point along the trace, $\rho_i (t_i) > \rho_{i+1} ( t_{i+1})$. 
	Thus, $(\rho_i (t_i))_{i \geq n}$ is a monotone decreasing sequence of natural numbers that does not converge, contradicting the fact that $\Nat$ is well-ordered.
\end{proof}
\noindent
Later, in Sect.~\ref{sect:ind-sim-cyc}, we will use the fact that the choices for generating an invalid branch in the proof above can be made \emph{uniformly} in an arithmetic setting.

\subsection{Defining $\CSn{n}$}
Simpson proposes in \cite{Sim17:cyclic-arith} to study systems of cyclic proofs containing only $\Sin{}{n}$ formulae, and to compare such systems to $\ISn{n}$.
This is rather pertinent in light of the free-cut elimination result we stated, Thm.~\ref{thm:free-cut-elim}: any $\ISn{n}$-proof of a $\Sin{}{n}$-sequent can be assumed to contain just $\Sin{}{n}$ formulae (possibly at a non-elementary cost in proof size), whence the comparison.
However, in order to be able to admit routine derivations of more complex formulae, {e.g.\ the $\Sin{}{n+1}$ law of excluded middle or the universal closure of a $\Sin{}{n}$ sequent}, we will close this notion under \emph{logical consequence}.


\begin{defi}
	\label{dfn:cphi}
	Let $\Phi$ be a set of formulae closed under subformulae and substitution.
	$\CC{\Phi}$ is the first-order theory axiomatised by the universal closures of conclusions of cyclic proofs containing only $\Phi$-formulae.
\end{defi}
%
Notice that, by the free-cut elimination result, Thm.~\ref{thm:free-cut-elim}, and the subformula property, any $\CSn n$ proofs of $\Sin{}{n}$-sequents contain only $\Sin{}{n}$-sequents anyway, without loss of generality.
%
This more `robust' definition allows us to easily compare fragments of cyclic arithmetic.
For instance, we have the following:


\begin{prop}
	\label{prop:csn-eq-cpn}
	$\CSn n = \CPn n$, for $n \geq 0$.
\end{prop}
\begin{proof}
	For the left-right inclusion, replace each $\Sin{}{n}$ sequent $\vec p ,\Gamma \seqar \Delta$ with the sequent $\vec p, \compl \Delta \seqar \compl \Gamma$, where $\compl \Gamma$ and $ \compl \Delta $ contain the De Morgan dual formulae of $\Gamma $ and $ \Delta$ resp.\ and $\vec p$ exhausts the atomic formulae of the antecedent.
	Any traces will be preserved and the proof can be made correct by locally adding some logical steps.
	The converse implication is proved in the same way.
\end{proof}


%
%
%
%

%
%

Using a standard technique, e.g.\ from \cite{Bro06:phdthesis}, we also can rather simply show the following result, which we will later strengthen in Sect.~\ref{sect:cyc-sim-ind}:\footnote{A similar result was given in \cite{Sim17:cyclic-arith}, but that argument rather shows that $\ISn n \subseteq \CSn{n+1}$.}

\begin{prop}
	\label{prop:isn-in-csn}
	$\CSn n $ proves any $\Pin{}{n+1}$ theorem of $\ISn n$, for $n\geq 0$.
\end{prop}

\begin{proof}
	[Proof sketch]
	Suppose $\ISn n $ proves $\forall \vec x . \psi(\vec x)$ where $\phi$ is $\Sin{}{n}$. 
	Let $\pi$ be a free-cut free $\ISn n$ proof of $\psi(\vec a)$, so in particular contains only $\Sin{}n$ formulas, by the subformula property. 
	We may now construct a $\CSn n $ proof of $\phi(\vec a)$ by simply simulating every local inference step of $\pi$; the only nontrivial case is the induction rule:
	\[
	\vliinf{\ind}{}{\Gamma \seqar \phi(t), \Delta}{\Gamma \seqar \phi(0), \Delta}{ \Gamma, \phi(a) \seqar \phi(\succ a), \Delta}
	\]
	This is simulated by the following cyclic derivation (omitting some routine proof steps),
	\[
	\vlderivation{
		\vlin{\sub}{}{\Gamma \seqar \phi(t), \Delta}{
			\vliin{\cut}{\bullet}{\Gamma \seqar \phi(b), \Delta}{
				\vlin{=}{}{0=b, \Gamma \seqar \phi(b), \Delta}{
					\vlhy{\Gamma \seqar \phi(0), \Delta}
				}
			}{
				\vlin{}{}{0< b , \Gamma \seqar \phi(b), \Delta}{
					\vlin{}{}{b= \succ a, \Gamma \seqar \phi(b), \Delta}{
						\vliin{\cut}{}{\underline{a< b}, \Gamma \seqar \phi(\succ a), \Delta}{
							\vlin{\sub}{}{\Gamma \seqar \phi(a) , \Delta}{
								\vlin{\cut}{\bullet}{\Gamma \seqar \phi(b), \Delta}{\vlhy{\vdots}}
							}
						}{
							\vlhy{\Gamma, \phi(a) \seqar \phi(\succ a), \Delta}
						}
					}
				}
			}
		}
	}
	\]
	where we have written $\bullet$ to mark roots of identical subtrees.
	An infinite branch that does not have a tail in the proofs of the two premisses of $\ind$ must eventually loop on $\bullet$. Therefore it admits an infinitely progressing trace alternating between $a$ and $b$, with the progress point underlined above.
	Now the proposition follows by simple application  of $\rigrul \forall$. 
\end{proof}

%
%
%
Following Rmk.~\ref{rmk:traces-occurrences}, notice that, e.g.\ in the simulation of induction above, traces need not be connected in the graph of ancestry of a proof. This deviates from other settings where it is \emph{occurrences} that are tracked, rather than terms, e.g.\ in  \cite{BaeDouHirSau16:multl,BaeDouSau16:cut-elim,Dou17:multl}.
%

\subsection{B\"uchi automata: checking correctness of cyclic proofs}
\label{sect:sect:automata-prelims}
A cyclic preproof can be effectively checked for correctness by reduction to the inclusion of `B\"uchi automata', yielding a $\pspace$ bound. {As far as the author is aware, this is the best known upper bound, although no corresponding lower bound is known.
	As we will see later in Sect.~\ref{sect:ind-sim-cyc}, this is one of the reasons why we cannot hope for a `polynomial simulation' of cyclic proofs in a usual proof system, and so why elementary simulations are more pertinent.}

\begin{defi}
	A \dfntrm{nondeterministic B\"uchi automaton} (NBA) $\mathcal A$ is a tuple $(A, Q, \delta , q_0 , F)$ where:
	$A$ is a finite set, called the \dfntrm{alphabet},
	$Q$ is a finite set of \dfntrm{states},
	$\delta \subseteq (Q \times A ) \times Q$ is the \dfntrm{transition relation},
	$q_0 \in Q$ is the \dfntrm{initial} state,
	and
	$F\subseteq Q$ is the set of \dfntrm{final} or \dfntrm{accepting} states.
	%
	We say that $\mathcal A$ is \dfntrm{deterministic} (a DBA) if $\delta $ is (the graph of) a function $Q \times A \to Q$.
	A `word' $(a_i)_{i \geq 0} \in A^\omega $ is \dfntrm{accepted} or \dfntrm{recognised} by $\mathcal A$ if there is a sequence $(q_i)_{i \geq 0} \in Q^\omega$ such that:
	for each $i\geq 0$, $(q_i, a_i , q_{i+1}) \in \delta$,
	and
	for infinitely many $i\geq 0$ we have $q_i \in F$.
	We write $\lang (\mathcal A)$ for the set of words in $A^\omega$ accepted by $\mathcal A$.
\end{defi}

From a cyclic preproof $\pi$ we can easily define two automata, say $\mathcal A^\pi_b$ and $\mathcal A^\pi_t$,\footnote{These are rather called $B_p$ and $B_t$ respectively in \cite{Sim17:cyclic-arith}.} 
respectively accepting just the infinite branches and just the infinite branches with infinitely progressing traces. {See \cite{Sim17:cyclic-arith} for a construction of $\mathcal A^\pi_t$.}
We point out that $\mathcal A^\pi_b$ is essentially just the dependency graph of $\pi$ with all states final, and so is in fact deterministic;\footnote{Technically the transition relation here is not total, but this can be `completed' in the usual way by adding a non-final `sink' state for any outstanding transitions.} 
we will rely on this observation later in Sects.~\ref{sect:ind-sim-cyc}, \ref{sect:nonuniform-ind-sim-cyc} and \ref{sect:red-to-det}.
We now state the well-known `correctness criterion' for cyclic proofs:

\begin{propC}
	[\cite{Sim17:cyclic-arith}]
	\label{prop:reg-proog-correctness-condition}
	A cyclic preproof $\pi$ is a $\infty$-proof iff $\lang(\mathcal A^\pi_b) \subseteq \lang (\mathcal A^\pi_t) $.
\end{propC}




\section{A translation from $\ISn{n+1}$ to $\CSn n$, over $\Pin{}{n+1}$-theorems}
\label{sect:cyc-sim-ind}
We show in this section our first result, that cyclic proofs containing only $\Sin{}{n}$-formulae are enough to simulate $\ISn{n+1}$ over not-too-complex formulae:

\begin{thm}
	\label{thm:cyclic-sim-ind}
	$\ISn{n+1} \subseteq \CSn{n}$, over $\Pin{}{n+1}$ theorems, for $n\geq 0$.
\end{thm}
{One example of such logical power in cyclic proofs was given in \cite{Sim17:cyclic-arith}, in the form of $\CSn 1$ proofs of the totality of the Ackermann-P\'eter function. This already separates it from $\ISn 1$, which only proves the totality of the primitive recursive functions \cite{Par72:n-quant-ind}.}
%
To prove the theorem above, we will rather work in $\IPn{n+1}$, cf.~Prop.~\ref{prop:ipn-equals-isn}, since the exposition is more intuitive.
We first prove the following intermediate lemma.

\begin{lem}
	\label{lem:ind-to-cyc-trans}
	Let $\pi $ be a $\IPn{n+1}$ proof, containing only $\Pin{}{n+1}$ formulae, of a sequent,
	\begin{equation}
	\label{eqn:pitwo-conc}
	\Gamma, \forall x_1 . \phi_1 , \dots , \forall x_l . \phi_l \seqar \Delta, \forall y_1 . \psi_1 , \dots , \forall y_m . \psi_m 
	\end{equation}
	where $\Gamma, \Delta, \phi_i, \psi_j$ are $\Sin{}{n}$ and $x_i , y_j$ occur only in $ \phi_i, \psi_j$ respectively.
	Then there is a $\CSn n$ derivation $\lift \pi$ of the form:
	\[
	\toks0={.5}
	\vlderivation{
		\vltrf{\lift \pi}{\Gamma \seqar \Delta, \psi_1 , \dots , \psi_m }{\vlhy{}}{\vlhy{\left\{ \Gamma \seqar \Delta,  \phi_i  \right\}_{i\leq l}}}{\vlhy{}}{\the\toks0}
	}
	\]
	Moreover, no free variables of \eqref{eqn:pitwo-conc} occur as eigenvariables for $\lefrul \exists$, $\rigrul \forall$ or $\sub$ steps in $\lift \pi$.
\end{lem}
\begin{proof}
	We proceed by induction on the structure of $\pi$.
	Notice that we may assume that any $\Pin{}{n+1}$ formulae occurring have just a single outermost $\forall$ quantifier, by interpreting arguments as pairs and using G\"odel's $\beta$ functions. {(This introduces only cuts on formulae of the same form.)}
	%
	We henceforth write $\vec \phi$ for $\phi_1 , \dots , \phi_l$ and $\vec \psi$ for $\psi_1 , \dots , \psi_m$ and, as an abuse of notation, $\forall \vec x . \vec \phi $ and $\forall \vec y . \vec \psi$ for $\forall x_1 . \phi_1 , \dots , \forall x_l . \phi_l $ and $\forall y_1 . \psi_1 , \dots , \forall y_m . \psi_m $ respectively. {(Notice that this is a reasonable abuse of notation, since $\forall$s can be prenexed outside conjunctions and disjunctions already in pure FO logic.)}

	Propositional logical steps are easily dealt with, relying on invertibility and cuts, with possible structural steps.
	Importantly, due to the statement of the lemma, such steps apply to only $\Sin{}{n}$ formulae (recall the discussion at the end of Sect.~\ref{sect:prelims-pa-pt}).
	For instance, if $\pi$ extends a proof $\pi'$ by a $\cand$-left step,
	\[
	\vlinf{\lefrul \cand}{}{\Gamma , \chi_0 \cand \chi_1, \forall \vec x . \vec \phi \seqar \Delta , \forall \vec y . \vec \psi}{\Gamma , \chi_0, \chi_1, \forall \vec x . \vec \phi \seqar \Delta , \forall \vec y . \vec \psi}
	\]
	then we define $\lift \pi$ as,
	\[
	\toks0={.25}
	\vlderivation{
		\vlin{\lefrul\cand}{}{\Gamma, \chi_0\cand \chi_1 \seqar \Delta , \vec \psi}{
			\vltrf{ \lift{\pi'} }{ \Gamma, \chi_0 , \chi_1 \seqar \Delta , \vec \psi }{\vlhy{}}{
				\vlhy{
					\left\{
					\vlderivation{
						\vliin{\cut}{}{\Gamma, \chi_0 , \chi_1 \seqar \Delta, \phi_i }{
							\vliq{}{}{\chi_0 , \chi_1 \seqar \chi_0 \cand \chi_1}{\vlhy{}}	
						}{ \vlhy{\Gamma, \chi_0 \cand \chi_1 \seqar \Delta, \phi_i} }
					}
					\right\}_{i\leq l}
				}	
			}{\vlhy{}}{\the\toks0}	
		}
	}
	\]
	and if $\pi$ extends proofs $\pi_0$ and $ \pi_1$ by a $\cand$-right step,
	\[
	\vliinf{\rigrul \cand}{}{ \Gamma , \forall \vec x . \vec \phi \seqar \Delta , \chi_0 \cand \chi_1 , \forall \vec y . \vec \psi }{\Gamma , \forall \vec x . \vec \phi \seqar \Delta , \chi_0 , \forall \vec y . \vec \psi}{\Gamma , \forall \vec x . \vec \phi \seqar \Delta ,  \chi_1 , \forall \vec y . \vec \psi}
	\]
	then we define $\lift \pi $ as:
	%
	
	\[
	\toks0={.3}
	\vlderivation{
		\vlin{\rigrul \cand}{\forall j <2}{\Gamma \seqar \Delta, \chi_0 \cand \chi_1, \vec \psi}{
			\vltrf{\lift{\pi_j}}{ \Gamma \seqar \Delta, \chi_j, \vec \psi }{\vlhy{}}{\vlhy{
					\left\{
					\vlderivation{
						\vliin{\cut}{}{\Gamma \seqar \Delta, \chi_j, \phi_i}{\vlhy{\Gamma \seqar \Delta, \chi_0 \cand \chi_1, \phi_i}}{
							\vliq{}{}{ \chi_0 \cand \chi_1 \seqar \chi_j }{\vlhy{}}	
						}		
					}
					\right\}_{i \leq l}
			}}{\vlhy{}}{\the\toks0}	
		}
	}
	\]
	
	If $\pi$ extends a proof $\pi'$ by a thinning step,
	\[
	\vlinf{\wk}{}{ \Gamma', \Pi, \Gamma, \forall \vec x . \vec \phi  \seqar \Delta , \forall \vec y . \vec \psi , \Delta', \forall \vec z . \vec \chi}{\Gamma, \forall \vec x . \vec \phi \seqar \Delta , \forall \vec y . \vec \psi}
	\]
	where $\Gamma', \Delta', \vec \chi$ are $\Sin{}{n}$ and $\Pi$ is $\Pin{}{n+1}$, then we define $\lift \pi$ as: 
	\[
	\toks0={0.5}
	\vlderivation{
		\vlin{\wk}{}{\Gamma', \Gamma \seqar \Delta, \vec \psi , \Delta' , \vec \chi}{
			\vltrf{\Gamma',\lift{\pi'}, \Delta'}{ \Gamma', \Gamma \seqar \Delta, \vec \psi , \Delta'}{\vlhy{}}{\vlhy{\left\{ \Gamma',\Gamma \seqar \Delta, \phi_i , \Delta' \right\}_{i \leq l} } }{\vlhy{}}{\the\toks0}
		}
	}
	\]
	where $\Gamma', \lift{\pi'}, \Delta'$ is obtained from $\lift{\pi'}$ by prepending $\Gamma'$ and appending $\Delta'$ to each sequent. 
	For this we might need to rename some free variables in $\pi'$ so that eigenvariable conditions are preserved after the transformation; this does not affect the cedents $\Gamma, \Delta$ by the assumption from the inductive hypothesis.
	{Notice that we are simply ignoring the extra premisses due to $\Pi$.}

	If $\pi$ extends proofs $\pi_0 $ and $\pi_1$ by a $\cut$ step on a $\Pin{}{n+1}$ formula,
	\[
	\vliinf{\cut}{}{\Gamma , \forall \vec x . \vec \phi \seqar \Delta , \forall \vec y . \vec \psi }{\Gamma , \forall \vec x . \vec \phi \seqar \Delta , \forall \vec y . \vec \psi , \forall z . \chi}{\Gamma , \forall \vec x . \vec \phi , \forall z . \chi \seqar \Delta , \forall \vec y . \vec \psi}
	\]
	then we define $\lift \pi$ as:
	\[
	\toks0={0.4}
	\toks1={0.3}
	\vlderivation{
		\vlid{}{}{\Gamma \seqar \Delta , \vec \psi}{
			\vltrf{\lift{\pi_1}, \vec \psi}{ \Gamma \seqar \Delta , \vec \psi , \vec \psi }{
				\vlhy{\vlderivation{\vltrf{\lift{\pi_0}}{\Gamma \seqar \Delta, \vec \psi , \chi }{\vlhy{}}{\vlhy{ \left\{ \Gamma \seqar \Delta , \phi_i \right\}_{i \leq l} }}{\vlhy{}}{\the\toks0}}}
			}{
				\vlhy{}
			}{
				\vlhy{
					\left\{
					\vlderivation{
						\vlin{\wk}{}{\Gamma \seqar \Delta , \vec \psi , \phi_i  }{\vlhy{ \Gamma \seqar \Delta, \phi_i } }
					}
					\right\}_{i \leq l}
				}
			}{
				\the\toks1
			}
		}
	}
	\]
	The final dotted `contraction' step is implicit, since we treat cedents as sets.
	Again, 
	we might need to rename some variables in $\pi_1$.
	If instead the cut formula were $\Sin{}n$, say $\chi$, we would define $\lift \pi $ as:
	\[
	\toks0={0.55}
	\toks1={0.4}
	\vlderivation{
		\vliin{\cut}{}{ \Gamma \seqar \Delta , \vec \psi }{
			\vltrf{\lift{\pi_0}}{ \Gamma \seqar \Delta , \chi , \vec \psi }{\vlhy{}}{
				\vlhy{\{ \Gamma \seqar \Delta , \phi_i \}_{i\leq l}}
			}{\vlhy{}}{\the\toks0}
		}{
			\vltrf{\lift{\pi_1}}{ \Gamma, \chi \seqar \Delta , \vec \psi}{\vlhy{}}{
				\vlhy{\left\{
					\vlderivation{
						\vlin{\wk}{}{\Gamma, \chi \seqar \Delta, \phi_i}{\vlhy{\Gamma \seqar \Delta , \phi_i}}	
					}
					\right\}_{i \leq l}}
			}{\vlhy{}}{\the\toks1}
		}
	}
	\]

	If $\pi$ extends a proof $\pi'$ by a $\forall$-left step,
	\[
	\vlinf{\lefrul \forall}{}{\Gamma , \forall z. \chi(z), \forall \vec x . \vec \phi \seqar \Delta, \forall \vec y. \vec \psi}{\Gamma , \chi(t), \forall \vec x . \vec \phi \seqar \Delta, \forall \vec y. \vec \psi}
	\]
	where $\forall z . \chi(z)$ is $\Pin{}{n+1}$, we define $\lift \pi$ as follows:
	\[
	\toks0={0.3}
	\vlderivation{
		\vliin{\cut}{}{\Gamma \seqar \Delta, \vec \psi}{
			\vlin{\sub}{}{\Gamma \seqar \Delta , \chi(t)}{\vlhy{\Gamma \seqar \Delta, \chi(z)}}
		}{
			\vltrf{\lift{\pi'}}{\Gamma , \chi(t) \seqar \Delta, \vec \psi}{\vlhy{}}{
				\vlhy{
					\left\{ \vlinf{\wk}{}{\Gamma, \chi(t) \seqar \Delta , \phi_i}{\Gamma \seqar \Delta, \phi_i} \right\}_{i\leq l}
				}
			}{\vlhy{}}{\the\toks0}
		}
	}
	\]
	(Notice that, although $z$ occurs as an eigenvariable for a $\sub$ step here, it is already bound in the conclusion of $\pi$, so we preserve the inductive hypothesis.)
	If $\pi$ extends a proof $\pi'$ by a $\forall$-right step,
	\[
	\vlinf{\rigrul \forall}{}{\Gamma , \forall \vec x . \vec \phi \seqar \Delta, \forall \vec y . \vec \psi , \forall z . \chi}{\Gamma , \forall \vec x . \vec \phi \seqar \Delta, \forall \vec y . \vec \psi , \chi}
	\]
	where $\forall z .\chi$ is $\Pin{}{n+1}$, then we define $\lift \pi$ as:
	\[
	\toks0={0.5}
	\vlderivation{
		\vltrf{\lift{\pi'}, \chi}{ \Gamma \seqar \Delta , \vec \psi , \chi }{\vlhy{}}{ 
			\vlhy{\left\{\vlderivation{\vlin{\wk}{}{\Gamma \seqar \Delta , \phi_i , \chi}{\vlhy{\Gamma \seqar \Delta, \phi_i}} } \right\}_{ i \leq l} }	
		}{\vlhy{}}{\the\toks0}
	}
	\]
	%
	%
	If $\pi$ extends a proof $\pi'$ by a $\exists$-right step,
	\[
	\vlinf{\rigrul \exists}{}{ \Gamma , \forall \vec x . \vec \phi \seqar \Delta, \forall \vec y . \vec \psi , \exists z . \chi(z) }{ \Gamma , \forall \vec x . \vec \phi \seqar \Delta, \forall \vec y . \vec \psi , \chi(t) }
	\]
	where $\exists z . \chi(z)$ is $\Sin{}{n}$, then we define $\lift \pi$ as:
	\[
	\toks0={0.3}
	\vlderivation{
		\vlin{\rigrul \exists}{}{ \Gamma \seqar \Delta, \vec \psi , \exists z . \chi(z) }{
			\vltrf{\lift{\pi'}, \exists z . \chi(z)}{ \Gamma \seqar \Delta , \vec \psi , \chi(t) , \exists z . \chi(z) }{\vlhy{}}{
				\vlhy{
					\left\{
					\vlinf{\wk}{}{\Gamma \seqar \Delta , \phi_i , \chi(t) , \exists z . \chi(z)}{\Gamma \seqar \Delta , \phi_i , \exists z . \chi(z)}
					\right\}_{i\leq l}	
				}
			}{\vlhy{}}{\the\toks0}
		}
	}
	\]	
	Again, some eigenvariables of $\pi'$ might have to be renamed.
	Any other quantifier steps are dealt with routinely.
	
	Finally, if $\pi$ extends proofs $\pi_0$ and $\pi'$ by an induction step,
	\[
	\vliinf{\ind}{}{\Gamma , \forall \vec x . \vec \phi \seqar \Delta, \forall \vec y . \vec \psi , \forall z . \chi(t)}{\Gamma , \forall \vec x . \vec \phi \seqar \Delta, \forall \vec y . \vec \psi , \forall z . \chi(0)}{ \Gamma , \forall \vec x . \vec \phi , \forall z . \chi(c) \seqar \Delta, \forall \vec y . \vec \psi , \forall z . \chi(\succ c) }
	\]
	we define $\lift \pi$ to be the following cyclic proof,
	\[
	\toks0={0.3}
	\toks1={0.5}
	\vlderivation{
		\vlin{\sub}{}{\Gamma \seqar \Delta, \vec \psi , \chi(t) }{
			\vliin{}{\daemon}{\Gamma \seqar \Delta ,  \vec \psi , \chi(d)}{
				\vlin{}{}{d=0, \Gamma \seqar \Delta , \vec \psi , \chi(d)}{
					\vltrf{\lift{\pi_0}}{\Gamma \seqar \Delta, \vec \psi , \chi(0)}{\vlhy{}}{\vlhy{ \left\{ \Gamma \seqar \Delta, \phi_i  \right\}_{i\leq l} }}{\vlhy{}}{\the\toks1}
				}
			}{
				\vlin{}{}{\underline{c<d}, d = \succ c , \Gamma \seqar \Delta, \vec \psi , \chi(d)}{
					\vltrf{\lift{\pi'},\vec \psi }{ \Gamma \seqar \Delta , \vec \psi , \chi(\succ c)}{
						\vlin{\sub}{}{\Gamma \seqar \Delta, \vec \psi, \chi(c)}{
							\vlin{}{\daemon}{\Gamma \seqar \Delta ,\vec \psi , \chi(d) }{\vlhy{\vdots}}
						} 
					}{
						\vlhy{}
					}{
						\vlhy{\left\{ \vlinf{\wk}{}{\Gamma \seqar \Delta, \vec \psi , \phi_i}{\Gamma \seqar \Delta , \phi_i} \right\}_{i\leq l} }
					}{
						\the\toks0
					}
				}
			}
		}
	}
	\]
	where we have written $\daemon$ to mark roots of identical subtrees.	
	Notice that any branch hitting $\daemon$ infinitely often will have an infinitely progressing trace alternating between $c$ and $d$, by the underlined progress point $c< d$: thanks to the assumption from the inductive hypothesis, $c$ will not occur in $\lift{\pi'}$ as an eigenvariable for $\lefrul \exists$, $\rigrul \forall$ or $\sub$ steps so the trace along $c$ in $\lift{\pi ' }$ remains intact, cf.~Rmk.~\ref{rmk:traces-occurrences}.
	Any other infinite branch has a tail that is already in $\lift{\pi'}$ or $\lift{\pi_0}$ and so has an infinitely progressing trace by the inductive hypothesis.
\end{proof}

The lemma above gives us a simple proof of the main result of this section:

\begin{proof}
	[Proof of Thm.~\ref{thm:cyclic-sim-ind}]
	Let $\pi$ be a $\IPn{n+1} $ proof of a sequent $\seqar \forall x . \phi$, where $\phi \in \Sin{}n$, under Prop.~\ref{prop:ipn-equals-isn}.
	By Thm.~\ref{thm:free-cut-elim} we may assume that $\pi$ contains only $\Pin{}{n+1}$ cuts, whence we may simply apply Lemma~\ref{lem:ind-to-cyc-trans} to obtain
	a $\CSn n $ proof of $\seqar \phi$.
	(Notice that there are no assumption sequents after applying the lemma since the antecedent is empty.)
	Now the result follows simply by an application of $\rigrul \forall$.
\end{proof}

For the interested reader, we have given an example of this translation in action in App.~\ref{sect:php-case-study}, for a `relativised' version of arithmetic with an uninterpreted function symbol.

\section{Second-order theories for reasoning about automata}
\label{sect:so-theories}
%
We now consider a two-sorted, or `second-order' (SO), version of FO logic, with variables $X,Y,Z, $ etc.\ ranging over sets of individuals, and new atomic formulae $t \in X$, sometimes written $X(t)$.
We also have SO quantifiers binding the SO variables with the natural interpretation.
Again, we give only brief preliminaries, but the reader is encouraged to consult the standard texts \cite{Sim09:reverse-math} and \cite{Hir14:reverse-math}.

We write $\Q_2$ for an appropriate extension of $\Q$ by basic axioms governing sets (see, e.g., \cite{Sim09:reverse-math} or \cite{Hir14:reverse-math}), and write $\Sin{0}{n}$ and $\Pin{0}{n}$ for the classes $\Sin{}{n}$ and $\Pin{}{n}$ respectively, but now allowing free set variables to occur.
%
%

\begin{defi}
	The \textbf{recursive comprehension} axiom schema is the following:\footnote{Notice that there is an unfortunate coincidence of the notation $\CA$ for `comprehension axiom' and `cyclic arithmetic', but the context of use should always avoid any ambiguity.}
	\[
	\CCA{\Din 0 1}
	\ :\  
	\forall \vec y , \vec Y . (\forall x . ( \phi(x,\vec y , \vec Y) \ciff \cnot \psi(x, \vec y , \vec Y ) ) \cimp \exists X . \forall x. (X(x) \ciff \phi (x)))
	\]
	where $\phi,\psi$ are in $\Sin 0 1 $ and $X$ does not occur free in $\phi$ or $\psi$.
	From here, the theory $\RCA$ is defined as $\Q_2 + \CCA{\Din 0 1} + \CIND{\Sin 0 1 }$.
	%
	%
	%
\end{defi}

Since we will always work in extensions of $\RCA$, which proves the totality of primitive recursive functions, we will conservatively add function symbols for primitive recursive functions on individuals whenever we need them.
%
%
We will also henceforth consider FO theories extended by `oracles', i.e.\ uninterpreted set/predicate variables, in order to access `uniform' classes of FO proofs. We write $\ISn n(X)$ for the same class of proofs as $\ISn n$ but where $X$ is allowed to occur as a predicate symbol. The usefulness of a $\ISn{n}(X)$ proof is that we may later substitute $X$ for a FO formula, say $\phi(-) \in \Din{}{m+1}$, to arrive at a $\ISn{m+n}$ proof of size $O(|\phi|)$. This `parametrisation' of a FO proof allows us to avoid unnecessary blowups in proof size induced by `non-uniform' translations from second-order theories; we implicitly use this observation for proof complexity bounds later, particularly in Sect.~\ref{sect:ind-sim-cyc}.

\smallskip

The following result is an adaptation of well known conservativity results, e.g.\ as found in \cite{Sim09:reverse-math,Hir14:reverse-math}, but we include a proof anyway for completeness.

\begin{prop}
	\label{prop:conservativity}
	$\RCA + \CIND{\Sin 0 n}$ is conservative over $\ISn n (X)$.
\end{prop}
\begin{proof}
	[Proof sketch]
	First we introduce countably many fresh set symbols $X^{\vec t, \vec Y}_{\phi, \psi}$, indexed by $\Sin 0 1$ formulae $\phi(x, \vec x, \vec X), \psi(x, \vec x , \vec X)$ with all free variables indicated, FO terms $\vec t$ with $|\vec t | = | \vec x|$ and SO variables $\vec Y$ with $|\vec Y| = | \vec X|$. 
	These will serve as witnesses to the sets defined by comprehension.
	We replace the comprehension axioms by initial sequents of the form:
	\begin{equation}
	\label{eqn:ca-init-pos}\vlinf{}{}{\Gamma, \forall x . ( \phi(x,\vec t , \vec Y) \ciff \cnot\psi(x, \vec t , \vec Y ) ),\phi(t,\vec t , \vec Y) \seqar  t \in X^{\vec t, \vec Y}_{\phi, \psi} , \Delta }{}
	\end{equation}
	\begin{equation}
	\label{eqn:ca-init-neg}
	\vlinf{}{}{\Gamma, \forall x . ( \phi(x,\vec t , \vec Y) \ciff \cnot\psi(x, \vec t , \vec Y ) ), t \in X^{\vec t, \vec Y}_{\phi, \psi}  \seqar \phi(t,\vec t , \vec Y) ,\Delta }{}
	\end{equation}
	
	\medskip
	\noindent
	It is routine to show that these new initial sequents are equivalent to the comprehension axioms for $\phi,\psi$.
	
	Now we apply free-cut elimination, Thm.~\ref{thm:free-cut-elim}, to a proof in such a system and replace every occurrence of $t \in X^{\vec t , \vec Y}_{\phi, \psi}$ with $\phi(t, \vec t , \vec Y)$,
	and every occurrence of $t \notin X^{\vec t , \vec Y}_{\phi,\psi}$ with $\psi(t, \vec t , \vec Y)$. {(Recall here that we assume formulae are in De Morgan normal form.)} 
	Any comprehension initial sequents affected by this replacement become 
 purely logical theorems.
	%
	Furthermore, any induction formulae remain $\Sin{0}{n}$, provably in pure logic, thanks to our consideration of whether $ X^{\vec t , \vec Y}_{\phi, \psi}$ occurs positively or negatively.
	Any extraneous free set variables in induction steps (except $X$), e.g.\ $Y$, may be safely dealt with by replacing any atomic formula $Y(s)$ with $\top$.
	The resulting proof is in $\ISn n (X)$.
\end{proof}
It is worth pointing out that, in general, the transformation from a SO proof to a FO proof can yield a possibly non-elementary blowup in the size of proofs, due to, e.g., the application of (free-)cut elimination.

\subsection{Formalisation of B\"uchi acceptance}
From now on we will be rather informal when talking about finite objects, e.g.\ automata, finite sequences, or even formulae.
In particular we may freely use such meta-level objects within object-level formulae when, in fact, we are formally referring to their `G\"odel numbers'.
%
Also, statements inside quotations, ``-'', will usually be (provably) recursive in any free variables occurring, i.e.\ $\Din 0 1 $.
This way quantifier complexity is (usually) safely measured by just the quantifiers outside quotations.

We often treat a set symbol $X$ as a binary predicate by interpreting its argument as a pair and using G\"odel's `$\beta$ functions' to primitive-recursively extract its components.
We use such predicates to encode sequences by interpreting $X(x,y)$ as ``the $x$\textsuperscript{th} symbol of $X$ is $y$''; this interpretation presumes we already have the totality and determinism of $X$ as a binary relation.
Formally, for a set $S$ and a set symbol $X$ treated as a binary predicate,
we will write $X \in S^\omega$ for the conjunction of the following two formulae,
\begin{equation}
\label{eqn:tot-oracle}
\forall x . \exists y \in S . X(x,y)
\end{equation}
\begin{equation}
\label{eqn:det-oracle}
\forall x , y , z  . ((X(x,y) \cand X(x,z)) \cimp y=z)
\end{equation}
i.e.\ $X$ is, in fact, the graph of a function $\Nat \to S$.
When we know that these formulae hold true for $X$, we may construe the expression $X(x)$ as a term in formulae, for instance writing $\phi(X(x))$ as shorthand for $\exists y . (X(x,y) \cand \phi(y))$ or, equivalently, $\forall y . (X(x,y) \cimp \phi(y))$.



\begin{defi}
	[Language membership]
	\label{dfn:lang-memb-so}
	Let $\mathcal A  = ( A, Q, \delta, q_0 , F )$ be a NBA and treat $X$ as a binary predicate symbol.
	We define the formula
	$X \in \lang (\mathcal A)$ as:
	\begin{equation}
	\label{eqn:nd-memb}
	X \in A^\omega  \ \cand \ \exists Y \in Q^\omega . 
	\left(
	\begin{array}{rl}
	& Y(0, q_0) \\
	\cand & \forall x . \ (Y(x), X(x) , Y(\succ x) ) \in \delta \\
	\cand & \forall x . \exists x' > x . \ Y(x') \in F
	\end{array}
	\right)
	\end{equation}
	If $\mathcal A$ is deterministic and $X \in A^\omega$, we write $q_X (x,y)$ for ``$y$ is the $x$\textsuperscript{th} state of the run of $X$ on $\mathcal A$'', which is provably recursive in $\RCA$.
	Similarly to 
	before,
	 we may write $\phi (q_X(x))$ as shorthand for $\exists y . (q_X(x,y) \cand \phi (y))$ or, equivalently in $\RCA$, for $\forall y . (q_X(x,y) \cimp \phi(y) )$.
	For DBA, we alternatively define $X \in \lang (\mathcal A)$ as:
	\begin{equation}
	\label{eqn:det-acc}
	X \in A^\omega
	\ \cand \ 
	\forall x . \exists x'> x .\  q_X(x') \in F
	\end{equation}
\end{defi}

\noindent
This `double definition' will not be problematic for us, since $\RCA$ can check if an automaton is deterministic or not and, if so, even prove the equivalence between the two definitions:

\begin{prop}
	$\RCA \proves \forall \text{ DBA } \mathcal A . (\eqref{eqn:nd-memb} \ciff \eqref{eqn:det-acc})$.
\end{prop}
\begin{proof}
	[Proof sketch]
	Let $\mathcal A = (A, Q , \delta, q_0, F )$ be a deterministic automaton.
	For the left-right implication let $Y\in Q^\omega$ be an `accepting run' of $X$ on $\mathcal A$ and use induction to show that $Y(x, q_X(x))$.
	For the right-left implication, we use comprehension to define an `accepting run' $Y\in Q^\omega$ by:
	\(
	Y(x,q)
	\ \ciff \ 
	q_X (x,q)
	\).
	Clearly the definition of $Y$ is $\Din 0 1$, and we can show that such $Y$ is a `correct run' by induction on $x$.
\end{proof}

Notice that, for a deterministic automaton, the formula for acceptance is \emph{arithmetical} in $X$, i.e.\ there are no SO quantifiers.
This will be rather important for uniformity in the simulation of cyclic proofs in the next section.

\subsection{Formalisations of some automaton constructions}

Recall that we may freely add symbols for primitive recursive functions to our language.
Since we rely on various results from \cite{KMPS16:buchi-reverse} as the `engine' behind some of our proofs, we will use their notions for manipulating automata.

For NBA $\mathcal A , \mathcal A'$, we write $\mathcal A^c$ and $\mathcal A \sqcup \mathcal A'$ to denote the complement and union constructions of automata from \cite{KMPS16:buchi-reverse} (Sects.~5 and 6 resp.).
%
We also write $\Empty(\mathcal A)$ for the recursive algorithm from \cite{KMPS16:buchi-reverse} (Sect.~6), expressed as a $\Sin{}1$ formula in $\mathcal A$, determining whether $\mathcal A$ computes the empty language.
It will also be useful for us later, in order to bound logical and proof complexity, to notice that DBA can already be complemented in $\RCA$. This is a rather unsurprising result but does not appear in \cite{KMPS16:buchi-reverse}, so we give it here.

For a DBA $\mathcal A = (A, Q , \delta , q_0 , F)$, we define a complementary NBA $ \mathcal A^c $ as follows, 
$$\mathcal A^c \ \dfn \ (A, (Q \times \{0\}) \cup ((Q\setminus F) \times \{1\}) , \delta^c ,(q_0 , 0), (Q \setminus F) \times \{1\} )$$ 
where $\delta^c \subseteq (Q^c \times A) \times Q^c $ (writing $Q^c$ for $Q \times \{0\} \cup (Q\setminus F) \times \{1\}$) is defined as:
\[
\begin{array}{rl}
& \{ ((q,0), a , (q',0)) \ : \ (q,a,q') \in \delta \} \\
\cup & \{ ( (q,i), a , (q',1) ) \ : \ (q,a,q') \in \delta, i = 0,1 ,  q' \in Q\setminus F  \}
\end{array}
\]
The idea behind this construction is that a run of $\mathcal A^c$ follows $\mathcal A$ freely for some finite time (in the `$0$' component), after which it may no longer visit final states of $\mathcal A$ (once in the `$1$' component). The determinism of $\mathcal A$ guarantees that such a word is not accepted by it.

By directly inspecting the definitions from \cite{KMPS16:buchi-reverse}, and DBA complementation above, we have the following properties:

\begin{obs}
	\label{obs:complexity-aut-constructions}
	For NBA $\mathcal A, \mathcal A'$ we have that:
	\begin{enumerate}
		\item $\Empty(\mathcal A)$ is a polynomial-time predicate in $\mathcal A$.
		\item $\mathcal A \sqcup \mathcal A' $ is constructible in polynomial-time from $\mathcal A$ and $\mathcal A'$.
		\item $\mathcal A^c$ is constructible in exponential-time from $\mathcal A$.
	\end{enumerate}
	For a DBA $\mathcal A$, we have that:
	\begin{enumerate}
		\setcounter{enumi}{3}
		\item $\mathcal A^c$ is constructible in polynomial-time from $\mathcal A$.
	\end{enumerate}
\end{obs}

\noindent
None of these bounds are surprising, due to known bounds on the complexity of union, complementation and emptiness checking for (non)deterministic B\"uchi automata.
Nonetheless it is important to state them for the particular constructions used in this work for bounds on proof complexity later.

\begin{lem}
	\label{lem:aut-clos-props-in-so-arith}
	From \cite{KMPS16:buchi-reverse} we have the following:
	\begin{enumerate}
		\item\label{item:emptiness-rca} $ \RCA\proves \forall \text{ NBA }\mathcal A . (\Empty (\mathcal A) \ciff \forall X \in A^\omega. X \notin \lang (\mathcal A))$.
		\item\label{item:union-rca} $ \RCA\proves \forall \text{ NBA } \mathcal A_1 , \mathcal A_2 . (X \in \lang (\mathcal A_1 \sqcup \mathcal A_2) \ciff (X \in \lang (\mathcal A_1)  \corr X \in \lang (\mathcal A_2) ))$. 
		\item\label{item:compl-rca-s2ind} $\RCA + \CIND{\Sin 0 2} \proves  \forall \text{ NBA } \mathcal A . ( X \in A^\omega \cimp (X \in \lang (\mathcal A^c) \ciff X \notin \lang (\mathcal A)))$.
	\end{enumerate}
	We also have that:
	\begin{enumerate}
		\setcounter{enumi}{3}
		\item\label{item:compl-dba-rca}	$\RCA \proves \forall \text{ DBA } \mathcal A .\ ( X\in A^\omega \cimp  (X \in \lang (\mathcal A^c) \ciff X \notin \lang (\mathcal A ) )) $.
	\end{enumerate}
\end{lem}
\begin{proof}
	\ref{item:emptiness-rca}, \ref{item:union-rca} and \ref{item:compl-rca-s2ind} follow from \cite{KMPS16:buchi-reverse}, namely from Prop.~6.1 and Lemma~5.2, 
	so we give a proof of \ref{item:compl-dba-rca}.
	
	Working in $\RCA$, let $\mathcal A = (A, Q , \delta , q_0 , F)$ be a DBA. 
	For the right-left implication,
	if $X \notin \lang (\mathcal A)$ then $\exists x . \forall x' > x .\  q_X(x) \notin F$, so let $x_0$ witness this existential.
	Now, define by comprehension the run $Y \in (Q^c)^\omega$ as follows:
	\[
	Y(x,y)
	\ \ciff \ 
	((x \leq x_0 \cand y = (q_X(x), 0 )) \corr ( x > x_0 \cand y = (q_X (x), 1) ) )
	\]
	Now, indeed $Y(0, (q_0, 0))$, since $q_X (0) = q_0$, and $Y$ is a correct run of $X$ on $\mathcal A^c$ by considering separately the cases $x< x_0$, $x=x_0$ and $x> x_0$. Finally, for any $x$, $Y$ hits a final state at $\max (x, x_0) + 1 > x$.

	For the left-right implication, suppose $X \in \lang (\mathcal A^c)$ and let $Y\in  (Q^c)^\omega$ be an accepting run.
	By induction we have $\forall x . (Y(x) = (q_X (x), 0) \corr Y(x) = (q_X (x), 1) )$.
	Now, $Y$ must eventually hit an accepting state of $\mathcal A^c$, i.e.\ in the $1$-component, say at position $x_0$.
	Again by induction, we may show that $Y$ remains in the $1$-component of $\mathcal A^c$ after $x_0$, and hence $q_X(x) \notin F$ for $x\geq x_0$, as required.
\end{proof}

\section{An exponential simulation of $\CA$ in $\PA$}
\label{sect:ind-sim-cyc}
In this section we will adapt Simpson's approach in \cite{Sim17:cyclic-arith} for showing that $\CA \subseteq \PA$ into a \emph{uniform} result in $\PA$. 
This essentially constitutes a formalisation of the soundness argument, Prop.~\ref{prop:sound-cyclic}, in a SO theory conservative over the target fragment of $\PA$.
The `uniformity' we aim for ensures that the possibly non-elementary blowup translating from SO proofs to FO proofs occurs once and for all for a single arithmetical theorem. 
Only then do we instantiate the theorem (inside $\PA$) by the cyclic proof in question, leading to an only elementary blowup.

To give an idea of how the result is obtained, and how our exposition refines that of \cite{Sim17:cyclic-arith}, we take advantage of the following aspects of the soundness argument for cyclic proofs:
\begin{enumerate}[label={(\alph*)}]
	\item\label{item:branch-aut-det} The B\"uchi automaton accepting all infinite branches of a cyclic proof is, in fact, deterministic, and so we can express acceptance of an $\omega$-word in this automaton arithmetically.
	\item\label{item:invalid-branch-uniform} A branch of invalid sequents and corresponding assignments, as in the proof of Prop.~\ref{prop:sound-cyclic}, can be uniformly generated from an initial unsatisfying assignment by an arithmetical formula.
	\item\label{item:finite-closure-ordinals} Since all inductions are only up to $\omega$, we need only \emph{arbitrarily often} progressing traces, rather than explicit infinitely progressing traces.
\end{enumerate}
\noindent
Together, these properties give us just enough `wiggle room' to carry out the soundness argument in a sufficiently uniform way.

Throughout this section we will also carefully track how much quantifier complexity is used in theorem statements, since we will later modify this argument to obtain a converse result to Thm.~\ref{thm:cyclic-sim-ind}. 


\subsection{An arithmetically uniform treatment of automata}
%
%
%
%
%

%
%

Referring to \ref{item:finite-closure-ordinals} above, we define an arithmetical corollary of NBA acceptance that is nonetheless sufficiently strong to formalise the soundness argument for cyclic proofs:
\begin{defi}
	[Arithmetic acceptance]
	\label{dfn:aracc}
	Let $\mathcal A = (A,Q,\delta, q_0, F)$ be a NBA and $X\in A^\omega$, and temporarily write:
	\begin{itemize}
		\item $F(x) \dfn $ ``$x$ is a finite run of $X$ on $\mathcal A$ ending at a final state''.
		\item $E(z,x,y  ) \dfn $ ``$z$ extends $x$ to a finite run of $X$ on $\mathcal A$ hitting $\geq y$ final states''
	\end{itemize}
	We define:
	\begin{equation}
	\label{eqn:aracc-dfn}
	\ArAcc (X, \mathcal{A})
	\  \dfn \  
	X\in A^\omega \cand 
	\exists x . 
	\left(
	F(x)
	\cand
	\forall y . \exists z . E(z,x,y)
	\right)
	\end{equation}
\end{defi}


\noindent
{For intuition, we may consider $\omega$-regular expressions rather than automata, which are of the form $\sum\limits_{i< n} e_i \cdot f_i^\omega$, for some $n\in \Nat$, without loss of generality. 
	The formula $\ArAcc$ for this expression essentially recognises infinite words that have prefixes of the form $\sigma\tau_k$ for some $\sigma \in \lang(e_i)$, for some $i<n$, and $\tau_k \in \lang (f_i^k)$ for each $k \in \Nat$.
	Clearly the condition $\ArAcc$ is a (provable) consequence of acceptance itself:
	
	\begin{prop}
		\label{prop:arith-acc}
		$\RCA \proves \forall \mathcal A . (X \in \lang (\mathcal A) \cimp \ArAcc (X, \mathcal A))$.
	\end{prop}
	
	\begin{proof}
		Working in $\RCA$, fix $\mathcal A = (A, Q , \delta, q_0, F)$ and suppose $X \in \lang(\mathcal A)$.
		Let $Y\in Q^\omega$ be an `accepting run' of $X$ on $\mathcal A$, cf.~\eqref{eqn:nd-memb}.
		We may show that,
		\begin{equation}
		\label{eqn:fin-pref-run}
		\exists z \in Q^* . 	\text{``$z$ is a finite prefix of $Y$ hitting $\geq y$ final states in $\mathcal A$''}
		\end{equation} 
		by $ {\Sin{0}{1}} $-induction on $y$, appealing to the unboundedness of final states in $Y$
		for both the base case and the inductive steps.
		Now, in the definition of $\ArAcc$ in \eqref{eqn:aracc-dfn}, we set $x$ to be the least such $z$ for  which \eqref{eqn:fin-pref-run}$[1/y]$ holds (again by induction), so that $F(x)$ from \eqref{eqn:aracc-dfn} holds.
		Thus, for any $y \in \Nat$, we may find an appropriate $z$ making $E(z,x,y)$ in \eqref{eqn:aracc-dfn} true by appealing to \eqref{eqn:fin-pref-run}.
		The fact that $z$ extends $x$ follows from leastness of $x$ and that $Y$ is a sequence, cf.~\eqref{eqn:tot-oracle} and \eqref{eqn:det-oracle}.	
	\end{proof}

	%
	
	Let us write $\mathcal A_1 \sqsubseteq \mathcal A_2$ for $\Empty ( ( \mathcal A_1^c \sqcup \mathcal A_2 )^c )$.
	We may now present our main `uniform' result needed to carry out our soundness proof in FO theories.
	\begin{thm}
		\label{thm:arithmetisation-of-correctness}
		$\RCA + \CIND{\Sin 0 2} $ proves:
		\begin{equation}
		\label{eqn:arith-form-prog-traces}
		\forall\ \text{DBA}\ \mathcal A_1 , \forall \ \text{NBA}\ \mathcal A_2 .\ 
		\left( 
		(\mathcal A_1 \sqsubseteq \mathcal A_2 \cand X \in \lang(\mathcal A_1) )
		\cimp 
		\ArAcc (X, \mathcal A_2 ) 
		\right)
		\end{equation}
		
	\end{thm}

	\begin{proof}
		Working in $\RCA + \CIND{\Sin 0 2}$, let $\mathcal A_1$ be a DBA and $\mathcal A_2$ be a NBA such that $X \in \lang (\mathcal A_1)$ and $\mathcal A_1 \sqsubseteq \mathcal A_2$.
		We have:
		\[
		\begin{array}{rll}
		& \Empty( (\mathcal A_1^c \sqcup \mathcal A_2)^c ) & \text{since $\mathcal A_1 \sqsubseteq \mathcal A_2$} \\
		\implies & \forall Y \in A^\omega.\  Y\notin \lang ( (\mathcal A_1^c \sqcup \mathcal A_2)^c ) & \text{by Lemma~\ref{lem:aut-clos-props-in-so-arith}.\ref{item:emptiness-rca}} \\
		\implies & \forall Y \in A^\omega.\  Y \in \lang(\mathcal A_1^c \sqcup \mathcal A_2) & \text{by Lemma~\ref{lem:aut-clos-props-in-so-arith}.\ref{item:compl-rca-s2ind}} \\
		\implies & \forall Y \in A^\omega. (Y \in \mathcal \lang(A_1^c) \corr Y \in \lang(\mathcal A_2)) & \text{ by Lemma \ref{lem:aut-clos-props-in-so-arith}.\ref{item:union-rca}}\\
		\implies & \forall Y \in A^\omega . ( Y \in \lang (\mathcal A_1)  \cimp Y \in \lang (\mathcal A_2)) & \text{by Lemma \ref{lem:aut-clos-props-in-so-arith}.\ref{item:compl-dba-rca}}\\
		\implies & X \in \lang (\mathcal A_2) & \text{since $X\in \lang(\mathcal A_1)$}\\
		\implies & \ArAcc(X, \mathcal A_2) & \text{by Prop.~\ref{prop:arith-acc}.} \qedhere
		\end{array}
		\]
	\end{proof}
	
	\noindent
	Noticing that DBA acceptance is also purely arithmetical in $X$ (cf.~\ref{item:branch-aut-det}), by the conservativity result Prop.~\ref{prop:conservativity}, we have:
	
	\begin{cor}\label{cor:is2-prov-aracc}
		$\ISn 2 (X)$ proves \eqref{eqn:arith-form-prog-traces}.
	\end{cor}
	

	\subsection{Formalising the soundness argument for cyclic proofs}
	At this point we are able to mostly mimic the formalisation of the soundness argument from \cite{Sim17:cyclic-arith}, although we must further show that a branch of invalid sequents, cf.~the proof of Prop.~\ref{prop:sound-cyclic}, is uniformly describable (cf.~\ref{item:invalid-branch-uniform}).
	%
	
	\smallskip
	
	For $n\geq 0$, let $\Nat, \rho \models_n \phi$ be an appropriate $\Din{}{n+1}$ formula (provably in $\ISn{n+1}$)
	asserting that a formula $\phi$ is true in $\Nat $ under the assignment $\rho$ of its free variables to natural numbers, as long as $\phi$ is a Boolean combination of $\Sin{}{n}$ (or $\Pin{}{n}$) formulae.\footnote{If $\phi$ is not a Boolean combination of $\Sin{}n$ formulae then $\Nat, \rho \models_n \phi$ crashes and returns $\bot$.}
	Formally, the formula $\Nat, \rho \models_n \phi$ takes as arguments the \emph{codes} of $\rho$ and $\phi$, i.e.\ their G\"odel numbers; the construction of such a formula for $\models_n$ is standard (see, e.g.,~\cite{Bus98:handbook-of-pt,Kay91:models-of-pa,HajPud:93}) and it has size polynomial in $n$.
	Importantly, there are $\ISn{n+1}$ proofs that $\models_n$ satisfies `Tarski's truth conditions'.
	Writing $\bool(\Phi)$ for the class of Boolean combinations of $\Phi$-formulae, we have:
	\begin{prop}
		[Properties of $\models_n$, see e.g.\ \cite{HajPud:93}]
		\label{prop:tarski}
		For $n\geq 0$, the following $\Pin{}{n+1}$ formulae have $\ISn{n+1}$ proofs of size polynomial in $n$:
		\begin{enumerate}
			\item\label{item:tarski-not} $\forall \phi \in \bool (\Sin{}n). \forall \rho. \ (\Nat, \rho \models_n \cnot \phi \ \ciff \ \Nat, \rho \notmodels_n \phi)$.
			\item\label{item:tarski-or} $\forall \phi,\psi \in \bool (\Sin{}n). \forall \rho. \ (\Nat, \rho \models_n (\phi \corr \psi) \ \ciff \ (\Nat, \rho \models_n \phi \ \corr \ \Nat, \rho \models \psi))$.
			\item\label{item:tarski-and} $\forall \phi,\psi \in \bool (\Sin{}n). \forall \rho. \ (\Nat, \rho \models_n (\phi \cand \psi) \ \ciff \ (\Nat, \rho \models_n \phi \ \cand \ \Nat, \rho \models \psi))$.
			\item\label{item:tarski-exists} $\forall \phi \in \Sin{}{n}. \forall \rho. \ ( \Nat, \rho \models_n \exists x . \phi \ \ciff \ \exists y . (\Nat, \rho \cup \{ x \mapsto y \} \models_n \phi ) )$.
			\item\label{item:tarski-forall} $\forall \phi \in \Pin{}{n}. \forall \rho. \ ( \Nat, \rho \models_n \forall x . \phi \ \ciff \ \forall y . (\Nat,  \rho \cup \{ x \mapsto y \} \models_n \phi ) )$.
		\end{enumerate}
		\smallskip 
		
		\noindent
		We also have $\ISn{n+1}$ proofs of size polynomial in $n$ of the \emph{substitution property}: 
		\begin{enumerate}
			\setcounter{enumi}{5}
			
			
			\item\label{item:tarski-subst} $\forall \phi \in \bool (\Sin{}n) . \forall \rho. \forall \text{ terms } t. \ ( \Nat, \rho \cup \{ a \mapsto \rho(t) \} \models_n \phi \ \ciff \ \Nat, \rho \models_n \phi[t/a] )$.
		\end{enumerate}
	\end{prop}
	
	\noindent
	In particular we have the \emph{reflection property}:
	\begin{prop}
		[Reflection]
		\label{prop:reflection}
		For $n\geq 0$ we have $\ISn{1} \proves \phi \ciff (\Nat, \emptyset \models_n \phi)$ with proofs of size polynomial in $n$ and $|\phi|$, for any closed formula $\phi \in \Sin{}n \cup \Pin{}n$.
	\end{prop}

	\noindent
	Henceforth, all our proof complexity bounds in $n$ follow from the fact that proofs are parametrised by $\models_n$ and its basic properties from Prop.~\ref{prop:tarski} above.
	
	\begin{defi}
		[Uniform description of an invalid branch]
		Let $\pi$ be a $\CA$ preproof of a sequent $\Gamma \seqar \Delta$, and let $n\in \Nat$ be such that all formulae occurring in $\pi$ are $\Sin{}{n}$. 
		Let $\rho_0$ be an assignment such that $\Nat, \rho_0 \models_n \bigwedge \Gamma$ but $\Nat ,  \rho_0 \notmodels_n \bigvee \Delta$.
		The branch of $\pi$ \dfntrm{generated} by $\rho_0$ is the invalid branch as constructed in the proof of Prop.~\ref{prop:sound-cyclic}, where at each step that there is a choice of premiss the leftmost one is chosen, and at each step when there is a choice of assignment of a natural number to a free variable the least one is chosen.
		We write $\Branch_n(\pi, \rho_0, x, y)$ for the following predicate:
		\[
		\text{``the $x$\textsuperscript{th} element of the branch generated by $\rho_0$ in $\pi$ is $y$''}
		\]
		To be precise, the `element' $y$ 
		is given as a pair $\pair{\rho_x}{\Gamma_x\seqar \Delta_x}$ consisting of 
		a sequent $\Gamma_x \seqar \Delta_x$ and an assignment $\rho_x$
		that invalidates it.
	\end{defi}

	Notice that $\Branch_n(\pi, \rho_0 , x , y)$ is recursive w.r.t.\ the oracle $\models_n$, and so is expressible by a $\Din{}{1} ( \models_n )$ formula, making it altogether $\Din{}{n+1} $ in its arguments.
	In fact, this is demonstrably the case in $\ISn{n+1}$, which can prove that $\Branch_n (\pi, \rho_0 , - , - )$ is the graph of a \emph{function}, as shown in Prop.~\ref{prop:inv-branch-isn} below.

	Let us
	write $\conc (\pi)$ for the conclusion of a $\CA$ proof $\pi$ and, as in Sect.~\ref{sect:sect:automata-prelims}, $\mathcal A^\pi_b$ and $\mathcal{A}^\pi_t$ for its branch and trace automata, resp.
	When we write $\Nat, \rho \models_n( \Gamma \seqar \Delta)$ we mean the $\Din{}{n+1}$ formula $(\Nat, \rho \notmodels_n \bigwedge\Gamma) \corr (\Nat, \rho \models_n \bigvee \Delta)$.
 
	\begin{prop}
		\label{prop:inv-branch-isn}
		%
		For $n\geq 0$, there are $\ISn{n+1}$ proofs of size polynomial in $n$ of:
		\begin{equation}
		\label{eqn:branch-is-total}
		\begin{array}{l}
		\forall  \pi \text{ a $\CA$ preproof containing only $\Sin{} n$ formulae}. \\
		\forall  \rho_0 . \ ((\Nat, \rho_0 \notmodels_n \conc (\pi) )\cimp \Branch_n (\pi, \rho_0 , - , -) \in \lang ( \mathcal A^\pi_b ))
		\end{array}
		\end{equation}
	\end{prop}
	\begin{proof}
		Working in $\ISn{n+1}$, let $\pi$ and $\rho_0$ satisfy the hypotheses of \eqref{eqn:branch-is-total} above.
		The fact that $\Branch_n(\pi, \rho_0 , - , - )$ is deterministic, cf.~\eqref{eqn:det-oracle}, follows directly by induction on the position of the branch.
		The difficult part is to show that $\Branch_n (\pi, \rho_0 , - , - )$ is total, cf.~\eqref{eqn:tot-oracle}, i.e.~that it never reaches a deadlock. For this we show,
		\begin{equation}
		\label{eqn:branch-total-invalid}
		\forall x . \exists \pair {\rho_x} {\Gamma_x\seqar \Delta_x} . \ (\Branch_n(\pi, \rho_0 , x , \pair {\rho_x} {\Gamma_x\seqar \Delta_x} ) \cand \ \Nat, \rho_x \notmodels_n (\Gamma_x\seqar \Delta_x) )
		\end{equation}
		by ${\Sin{}{n+1}}$-induction on $x$.
		The base case, when $x=0$, follows by assumption, so we proceed with the inductive case.
		For a given $x$ let $\pair{\rho_x}{\Gamma_x \seqar \Delta_x}$ witness \eqref{eqn:branch-total-invalid} above and let $\rul$ be the rule instance in $\pi$ that $\Gamma_x \seqar \Delta_x$ concludes.
		
		If $\rul$ is a $ \rigrul \exists$ step with associated term $t$, then there is only one premiss which we show remains false in the current assignment. This follows from:
		\[
		\begin{array}{rcll}
		\Nat, \rho \notmodels_n \exists x . \phi & \implies & \Nat, \rho \models_n \forall x . \cnot \phi & \text{by Prop.~\ref{prop:tarski}.\ref{item:tarski-not}} \\
		& \implies & \forall y . (\Nat, \rho \cup  \{ x \mapsto y \} \models_n \cnot \phi ) & \text{by Prop.~\ref{prop:tarski}.\ref{item:tarski-forall}} \\
		& \implies & \Nat, \rho \cup \{ x \mapsto \rho(t) \} \models_n  \cnot \phi & \text{by pure logic} \\
		& \implies & \Nat, \rho \models_n \cnot \phi [t/x] & \text{by Prop.~\ref{prop:tarski}.\ref{item:tarski-subst}} \\
		& \implies & \Nat, \rho \notmodels_n \phi[t/x] & \text{by Prop.~\ref{prop:tarski}.\ref{item:tarski-not}.}
		\end{array}
		\]
		
		If $\rul$ is a $\rigrul{\forall}$ step 
		then there is only one premiss, for which we show that the appropriate invalidating assignment exists. This follows from,
		\[
		\begin{array}{rcll}
		\Nat, \rho \notmodels_n \forall x . \phi & \implies & \Nat, \rho \models_n \exists x . \cnot \phi & \text{by Prop.~\ref{prop:tarski}.\ref{item:tarski-not}} \\
		& \implies & \exists y . ( \Nat, \rho \cup\{ x \mapsto y \} \models_n \cnot \phi ) & \text{by Prop.~\ref{prop:tarski}.\ref{item:tarski-exists}} \\
		& \implies & \exists \text{ least } y  . ( \Nat, \rho \cup\{ x \mapsto y \} \models_n \cnot \phi ) & \text{by $\CIND{\Sin{}{n+1}}$
		}\\
		& \implies & \exists \text{ least } y . ( \Nat, \rho\cup \{ x \mapsto y \} \notmodels_n \phi ) & \text{by Prop.~\ref{prop:tarski}.\ref{item:tarski-not}.}
		\end{array}
		\]
		where, in the penultimate implication, we rely on the fact that the appropriate `minimisation' property is provable in $\CIND{\Sin{}{n+1}}$ (see, e.g., \cite{Bus98:handbook-of-pt}).
		
		If $\rul$ is a left quantifier step then it is treated similarly to the two right quantifier cases above by De Morgan duality. If $\rul $ is a propositional step then the treatment is simple, following directly from Prop.~\ref{prop:tarski}. If $\rul$ is an initial sequent we immediately hit a contradiction, since all of the axioms are provably true in all assignments.
		If $\rul $ is a substitution step, then the existence of the appropriate assignment follows directly from the substitution property, Prop.~\ref{prop:tarski}.\ref{item:tarski-subst}.

		Finally, we may show that every state of the run of $\Branch_n (\pi, \rho_0 ,- , - )$ on $\mathcal A^\pi_b$ is final, by $\Sin{}{n+1}$-induction, since it always correctly follows a branch of $\pi$. Thus we have that $\Branch_n (\pi, \rho_0, -, - ) \in \lang (\mathcal A^\pi_b)$.
	\end{proof}
	
	

	Now we can give a formalised proof of the soundness of cyclic proofs:
	\begin{thm}
		[Soundness of cyclic proofs, formalised]
		\label{thm:soundness-formalised}
		For $n \geq 0$, there are $\ISn{n+2}$ proofs of size polynomial in $n$ of:
		\begin{equation}
		\label{eqn:soundness-polysize}
		\begin{array}{l}
		\forall \pi \text{ a $\CA$ preproof containing only $\Sin{} n$ formulae}. \\
		(\mathcal A^\pi_b \sqsubseteq \mathcal A^\pi_t
		\ \cimp\ 
		\forall \rho_0  . \ \Nat, \rho_0 \models_n \conc (\pi))
		\end{array}
		\end{equation}
	\end{thm}
	
	\begin{proof}
		First, instantiating $X$ in Cor.~\ref{cor:is2-prov-aracc} with a $\Din{}{n+1}$ formula $\phi_n$ yields $O(|\phi_n|)$-size $\ISn{n+2}$ proofs of \eqref{eqn:arith-form-prog-traces}$[\phi_n/X]$.
		Hence, setting $\phi_n$ to be $\Branch_n (\pi, \rho_0 , - , - )$ and appealing to Prop.~\ref{prop:inv-branch-isn} above, we arrive at $\ISn{n+2}$ proofs of size polynomial in $n$ of:
		\begin{equation}
		\label{eqn:rfn-princ}
		\begin{array}{l}
		\forall  \pi \text{ a $\CA$ preproof containing only $\Sin{} n$ formulae}. \\
		\mathcal A^\pi_b \sqsubseteq \mathcal A^\pi_t
		\cimp
		\forall \rho_0 .
		(\Nat, \rho_0 \notmodels_n \conc (\pi)
		\cimp 
		\ArAcc( \Branch_n (\pi, \rho_0 , - , - ) , \mathcal A^\pi_t )
		)
		\end{array}
		\end{equation}
		
		Now, working in $\ISn{n+2}$, to prove \eqref{eqn:soundness-polysize} let $\pi $ satisfy $\mathcal A^\pi_b \sqsubseteq \mathcal A^\pi_t$. 
		For contradiction
		assume, for some $\rho_0$, that $\Nat, \rho_0 \notmodels_n \conc (\pi)$.
		By Prop.~\ref{prop:inv-branch-isn} we have
		$\Branch_n (\pi, \rho_0 , - , - ) \in \lang (\mathcal A^\pi_b)$,
		so we henceforth write $\Gamma_x \seqar \Delta_x$ and $\rho_x$ for the sequent and assignment at the $x$\textsuperscript{th} position of $\Branch_n (\pi, \rho_0 , -, - )$.
		By \eqref{eqn:rfn-princ} above we have $\ArAcc( \Branch_n (\pi, \rho_0 , - , - ) , \mathcal A^\pi_t )$, so let $x$ witness its outer existential, cf.~\eqref{eqn:aracc-dfn}.
		Now, let $y$ be the maximum value of $\rho_x (t)$ for all terms $t$ occurring in $\Gamma_x \seqar \Delta_x$.
		Again by $\ArAcc( \Branch_n (\pi, \rho_0 , - , - ) , \mathcal A^\pi_t )$, we have that there is some (finite) trace $z$ beginning from $\Gamma_x \seqar \Delta_x $ that progresses $y+1$ times. 
		Writing $z(i)$ to denote the $i$\textsuperscript{th} term in the trace $z$, we may show by induction on $i\leq |z|$
		that, if there are $j$ progress points between $z(0)$ and $z(i)$, then we have that $\rho_x (z(0)) \geq \rho_{x+i} (z(i)) +j$.
		In particular, $y \geq \rho_x (z(0)) \geq \rho_{x + |z|} (z(|z|)) + (y +1) \geq y+1$, yielding a contradiction. 
	\end{proof}

	\subsection{$\PA$ exponentially simulates $\CA$}
	
	We can now give our main proof complexity result:
	\begin{thm}
		\label{thm:elementary-simulation}
		If $\pi $ is a $\CA$ proof of $\phi$, then we can construct a $\PA$ proof of $\phi$ of size exponential in $|\pi|$.
	\end{thm}
	\begin{proof}
		%
		Take the least $n\in\Nat$ such that $\pi$ contains only $\Sin{}{n}$ formulae; in particular $n\leq |\pi|$. 
		Since $\pi$ is a correct cyclic proof, there is a $\PA$ proof of $\mathcal A^\pi_b \sqsubseteq \mathcal A^\pi_t$, by exhaustive search.
		In fact, such a proof in $\Q$ may be constructed in exponential time in $|\pi|$, thanks to Obs.~\ref{obs:complexity-aut-constructions} and $\Sigma_1$-completeness of $\Q$ (see, e.g., \cite{HajPud:93}).
		Hence, by instantiating $\pi$ in Thm.~\ref{thm:soundness-formalised}, we have $\ISn{n+2}$ proofs of $\Nat, \emptyset \models_n \phi$ of size exponential in $|\pi|$. 
		Finally by the reflection property, Prop.~\ref{prop:reflection}, we have that $\ISn{n+2} \proves \phi$ with proofs of size exponential in $|\pi|$.
	\end{proof}
	%

	Notice that we already have a converse polynomial simulation of $\PA$ in $\CA$ by the results of \cite{Sim17:cyclic-arith} or, alternatively, by Prop.~\ref{prop:isn-in-csn}.
	
	\section{$\ISn{n+1}$ contains $\CSn n$}
	\label{sect:nonuniform-ind-sim-cyc}
	In fact the proof method we developed in the last section allows us to recover a result on logical complexity too. 
	By tracking precisely all the bounds therein, we obtain that $\CSn n$ is contained in $\ISn{n+2}$, which is already an improvement to Simpson's result (see Sect.~\ref{sect:conc} for a comparison).
	To derive such bounds, in this section we concern ourselves only with cyclic proofs containing $\Sin{}{n}$ formulae.
	The universal closures of the conclusions of such proofs axiomatise $\CSn{n}$, cf.~Dfn.~\ref{dfn:cphi}, so more complex theorems of $\CSn{n}$ are thence derivable by pure logic.
	
	In fact, we may actually improve this logical bound and arrive at an optimal result (given Thm.~\ref{thm:cyclic-sim-ind}).
	By more carefully analysing the proof methods of \cite{KMPS16:buchi-reverse}, namely an inspection of the proofs of Thms.~5 and 12 in that work, we have that:
	\begin{prop}
		[Implicit in \cite{KMPS16:buchi-reverse}]
		\label{prop:nonuniform-compl}
		$\RCA \proves \forall X \in A^\omega . (X\in \lang(\mathcal A^c) \ciff X \notin \lang (\mathcal A) )$, for any NBA $\mathcal A$.
	\end{prop}
	\noindent
	Notice here that the universal quantification over NBA is \emph{external}, so that the complementation proofs are not necessarily uniform. 
	This is not a trivial result, since it relies on a version of Ramsey's theorem, the \emph{additive Ramsey theorem}, which can be proved by induction on the number of `colours'. 
	Usual forms of Ramsey's theorem are not proved by such an argument, and in fact it is well known that $\RCA$ cannot even prove Ramsey's theorem for pairs with only two colours (see, e.g., \cite{Hir14:reverse-math}).
	We include in App.~\ref{sect:complementation-nonuniform} a self-contained (and somewhat simpler) proof of Prop.~\ref{prop:nonuniform-compl} above, for completeness.

	This allows us to `un-uniformise' the results of the previous section, using Prop.~\ref{prop:nonuniform-compl} above instead of Lemma~\ref{lem:aut-clos-props-in-so-arith}.\ref{item:compl-rca-s2ind}, in order to `trade off' proof complexity for logical complexity:
	
	
	\begin{prop}
		[Soundness of cyclic proofs, non-uniformly formalised]
		\label{prop:soundness-nonunif}
		Let $n\geq 0 $ and $\pi$ be a $\CA $ proof containing only $\Sin{}{n}$ formulae. 
		$\ISn{n+1} \proves \forall \rho_0.  (\Nat, \rho_0 \models_n \conc (\pi))$.
	\end{prop}
	\begin{proof}
		[Proof sketch]
		We mimic the entire argument of Thm.~\ref{thm:soundness-formalised} by instantiating the fixed proof $\pi$ and using Prop.~\ref{prop:nonuniform-compl} above instead of Lemma~\ref{lem:aut-clos-props-in-so-arith}.\ref{item:compl-rca-s2ind}. 
		In particular, the required `non-uniform' versions of Thm.~\ref{thm:arithmetisation-of-correctness} and Cor.~\ref{cor:is2-prov-aracc} become derivable in $\RCA$ and $\ISn 1 (X)$ resp., thus reducing the global induction complexity by one level.
		Hence we arrive at a `non-uniform' version of Thm.~\ref{thm:soundness-formalised}, peculiar to the fixed proof $\pi$ we began with, proved entirely within $\ISn {n+1}$, as required.
	\end{proof}

	\begin{thm}
		\label{thm:nonuniform-ind-sim-cyc}
		For $n\geq 0$, we have that $\CSn n \subseteq \ISn{n+1}$.
	\end{thm}
	
	\begin{proof}
		By the definition of $\CSn n$, cf.~Dfn.~\ref{dfn:cphi}, it suffices to derive in $\ISn{n+1}$ just the (possibly open) conclusions of $\CSn{n}$ proofs containing only $\Sin {} n$ formulae,
		and so the result follows directly from Props.~ \ref{prop:soundness-nonunif} and \ref{prop:reflection}.
		%
	\end{proof}
	
	\subsection{On the proof complexity of $\CSn n$}
	\label{sect:sect:prf-comp-csn}
	One might be tempted to conclude that the elementary simulation of $\CA$ by $\PA$ should go through already for $\CSn n$ by $\ISn{n+2}$ (independently of $n$), due to the bounds implicit in the proof of Thm.~\ref{thm:elementary-simulation}. 
	Furthermore, if we are willing to give up a few more exponentials in complexity, one may even bound the size of $\ISn 1$ proofs arising from Prop.~\ref{prop:nonuniform-compl} by an appropriate elementary function (though this analysis is beyond the scope of this paper).
	%
	
	However, we must be conscious of the `robustness' of the definition of $\CSn n$ proofs in terms of complexity.
	The one we gave, which essentially requires cyclic proofs to contain only $\Sin{}{n}$-formulae, is more similar to `free-cut free' $\ISn n $ proofs than general ones, cf.~Thm.~\ref{thm:free-cut-elim}, so it seems unfair to compare these notions of proof in terms of proof complexity.
	In fact we may define a more natural notion of a $\CSn n$ proof from the point of view of complexity, while inducing the same theory.
	
	\smallskip
	First, let us recall some notions from, e.g., \cite{Bro06:phdthesis,BroSim11:seq-calc-ind-inf-desc}. 
	Any cyclic proof can be written in `cycle normal form', where any `backpointer' (e.g., the conclusions of upper sequents marked $\bullet$ until now) points only to a sequent that has occurred below it in the proof.
	Referring to the terminology of \cite{Bro06:phdthesis,BroSim11:seq-calc-ind-inf-desc} etc., we say that a sequent is a \emph{bud} in a cyclic proof if it has a backpointer pointing to an identical sequent (called the \emph{companion}). 
	\begin{prop}
		\label{prop:cyclic-sn-general-dfn}
		If $\phi $ has a $\CA$ proof whose buds and companions contain only $\Sin {} n$ formulae then $\CSn n \proves \phi$.	
	\end{prop}
	\begin{proof}
		[Proof idea]
		Once again, we simply apply free-cut elimination, Thm.~\ref{thm:free-cut-elim}, treating any backpointers as `initial sequents' for a proof in cycle normal form.
	\end{proof}
	%
	
	This result again shows the robustness of the definition of $\CSn n$ as a theory, and we would further argue that, from the point of view of proof complexity, the backpointer-condition induced by Prop.~\ref{prop:cyclic-sn-general-dfn} above constitutes a better notion of `proof' for $\CSn n$. 
	One could go yet further and argue that an even better notion of `proof' would allow pointers to any other identical sequent in the proof (not just those below it), which could potentially make for an exponential improvement in proof complexity.

	At the same time we see that it is not easy to compare the proof complexity of such systems for $\CSn n $ with those for $\ISn{n+1}$, due to the fact that we have used a free-cut elimination result for the simulations in both directions, inducing a possibly non-elementary blowup in proof size.
	It would be interesting if a more fine-grained result regarding the relative proof complexity of $\ISn{n+1}$ and $\CSn{n}$ could be established, but this is beyond the scope of the current work.

	%
	

	\section{Some metamathematical results}
	\label{sect:further-metalogical}
	
	In this section we make some further observations on various properties of the cyclic theories in this work, which are later applied in Sect.~\ref{sect:red-to-det}. The exposition we give is brief, since we follow standard methods, but we provide appropriate references for the reader.

	\subsection{Provably recursive functions of $\CDn 0$}
	\label{sect:sect:prov-rec-fns-cd0}
	As a corollary of the Thms.~\ref{thm:cyclic-sim-ind} and \ref{thm:nonuniform-ind-sim-cyc} we have that, for $n\geq 1$, the provably recursive functions of $\CSn n $ coincide with those of $\ISn{n+1}$. Such functions were characterised by Parsons in \cite{Par72:n-quant-ind} as certain fragments of G\"odel's system $\mathsf T$ or, equivalently, by recursion up to suitably bounded towers of $\omega$, e.g.\ in \cite{Buss1995witness}. 
	This apparently leaves open a gap for the case of $\CDn 0 $ (a.k.a.\ $ \CSn 0$).
	However notice that from the definition of $\CSn n $ and Thm.~\ref{thm:cyclic-sim-ind} we have:
	
	\begin{cor}
		[of Thm.~\ref{thm:cyclic-sim-ind}]
		\label{cor:pi-axiomatised}
		$\CSn n$ is axiomatised by the $\Pin{}{n+ 1}$-consequences of $\ISn {n+1}$, for $n\geq 0$.
	\end{cor}
	
	\noindent
	In particular, since $\IDn 0 $ is known to be $\Pin{}1$-axiomatised (see, e.g., Thm.~1.27 in \cite{HajPud:93}), we have that in fact $\CDn 0 \supseteq \IDn 0 $ (i.e.\ over \emph{all} theorems).
	Thus $\CDn 0 $ can prove recursive \emph{at least} the functions (bitwise computable) in the linear-time hierarchy, from the analogous result for $\IDn 0$ (see, e.g., Theorem~III.4.8 in \cite{Cook:2010:LFP:1734064}).
	 Conversely, since $\CDn 0 $ is also $\Pin{}1$-axiomatised by the above observation, it has no more `computational content' than $\IDn 0 $, as we will now see.
	 
	 Recall that the {linear-time hierarchy} ($\mathbf{LTH}$) is the class of predicates expressible by a $\Din{}0$ formula in the language of arithmetic.\footnote{Equivalently it is the class of predicates recognised by an alternating Turing machine with random access in a linear number of steps with only boundedly many alternations. See, e.g., \cite{CloKra02} for more details.}
	 A function is said to be in $\flth$ just if it has linear growth rate (in terms of the bit string representation) and is bitwise computable in $\lth$. 
	
	\begin{prop}
		\label{prop:cd0-lth}
		If $\CDn 0 \proves \forall \vec x . \exists ! y . \phi(\vec x , y)$ for a $\Sin{}1$-formula $\phi$, then $\phi$ computes the graph of a function in $\flth$.
	\end{prop}

%
%

%

	\begin{proof}
		[Proof sketch]
%
Suppose $\phi $ is $\exists \vec y . \phi_0$ with $\phi_0 \in \Din{} 0 $.
			By `Parikh's theorem' for $\Pin{}1$-axiomatised theories (see, e.g., Thm.~III.2.3 in \cite{Cook:2010:LFP:1734064}) we moreover have $\CDn 0 \proves \forall \vec x . \exists! y < t . \exists \vec y < \vec t .  \phi_0$ for some terms $t, \vec t$.
		This means we may simply search for a witness $y<t$ of linear size verifying $\exists \vec y < \vec t . \phi_0 $, an $\lth$ property, and thus $\exists \vec y < \vec t . \phi_0 $ computes the graph of a function in $\flth$.
\end{proof}


Interestingly, we cannot strengthen the above proposition to ``$\CDn 0 $ and $\IDn 0 $ prove the same $\Pin{}2$ theorems''. This is because $\ISn 1$ proves the \emph{consistency} of $\IDn 0 $, a $\Pin{} 1$ sentence, so $\CDn 0 $ does too by Thm.~\ref{thm:cyclic-sim-ind}.
Thus the aforementioned stronger statement would contradict G\"odel's second incompleteness theorem for $\IDn 0 $.
	\smallskip
	
	See, e.g., \cite{Bus98:handbook-of-pt,Cook:2010:LFP:1734064} for further discussions on the provably recursive functions of fragments of (bounded) arithmetic, and see, e.g., \cite{CloKra02} for more details on relationships between the language of arithmetic and recursive function classes.

	\subsection{Failure of cut-admissibility}
	\label{sect:sect:cut-ad-fails}
	As a corollary of our results, we may formally conclude that the cut rule is not admissible in $\CA$, or indeed any of its fragments $\CSn n$.\footnote{This observation was pointed out to me by Stefano Berardi.}
	In fact the situation is rather worse than that:
	
	\begin{cor}
		[of Thms.~\ref{thm:cyclic-sim-ind} and \ref{thm:nonuniform-ind-sim-cyc}]
		\label{prop:cut-non-admiss}
		Let $n\geq 1$. The class of $\CA$ proofs with only $\Sin {} {n-1} $ cuts is not complete for even the $\Pin{}1$ theorems of $\CSn {n}$. 
	\end{cor}

	\begin{proof}
		For a recursively axiomatised theory $T$, let $\cons T$ be an appropriate $\Pin{}{1}$ sentence expressing that ``$T$ does not prove $0=1$''.
		It is well-known that $\ISn{n+1} \proves \cons {\ISn n}$ (see, e.g., \cite{Kay91:models-of-pa,Bus98:handbook-of-pt,HajPud:93}), so also $\CSn n \proves \cons {\ISn n}$ by Thm.~\ref{thm:cyclic-sim-ind}. 
		For contradiction, suppose $\cons{\ISn n }$ concludes some $\CA$ proof with only $\Sin{}{n-1}$ cuts; then in fact $\CSn{n-1} \proves \cons {\ISn n}$ by degeneralising (for the case $n=1$) and the subformula property. However this implies $\ISn n \proves \cons{\ISn n }$ by Thm.~\ref{thm:nonuniform-ind-sim-cyc}, which is impossible by G\"odel's second incompleteness theorem for $\ISn n $. 
	\end{proof}

	See, e.g., \cite{Bus98:handbook-of-pt,Kay91:models-of-pa,HajPud:93} for further discussions on the provability of consistency principles for fragments of arithmetic.

	\subsection{Reflection and consistency}
	\label{sect:sect:rfn-cons}
	
	Thanks to the uniformity of the results from Sect.~\ref{sect:ind-sim-cyc}, we can give some fundamental metalogical properties regarding provable soundness and consistency of cyclic proofs.
	First we will fix our formalisation of $\CSn n $-provability (of arbitrary sequents, not just $\Sin{}{n}$) in the language of arithmetic. 
	
	\smallskip
	Let $n\geq 0$. 
	We will fix some appropriate formula $\prf_n (\pi,\phi)$ expressing that $\pi$ is a $\CSn n $ proof of $\phi$. 
	We suppose that proofs are written as usual derivations (finite trees or dags)
	whose leaves are either axiom instances from $\Q$, or otherwise some (possibly open) $\Sin{}n$ sequent labelled by an associated cyclic proof that derives it (containing only $\Sin{}{n}$ formulae).
	Descriptively, $\prf_n (x,y)$ checks that $\pi$ is a proof of $\phi$ by first checking that it is a well-formed derivation, then checking that each premiss is either an axiom instance from $\Q$ or is labelled by a correct cyclic proof deriving it.
	In the latter case it must search for a certificate verifying that the cyclic proof satisfies the automaton-inclusion condition, i.e.\ that $\mathcal A_b \sqsubseteq \mathcal A_t$.
	%
	
	While $\prf_n (\pi,\phi)$ is recursive in $\pi$ and $\phi$, this may not be provably the case in weak theories such as $\IDn 0 $.
	%
	Thus we fix $\prf_n $ to be an appropriate $\Sin{}1$ formula, as described above,
	and we write $\Prov \phi$ for $\exists \pi . \prf_n (\pi, {{\phi}})$.
	%
	We write $\Rfn{\Pin{}{k}}{\CSn n }$ for the (local) \emph{$\Pin{}{k}$-reflection principle} of $\CSn n $. I.e.\
	\[
	\Rfn{\Pin{}{k}}{\CSn n }
	\  \dfn \ 
	\{\Prov \phi \cimp \phi : \phi \in \Pin{}{k} \}
	\]

	\begin{cor}
		[of Thm.~\ref{thm:soundness-formalised}]
		\label{cor:reflection}
		For $n\geq 0$, we have $\ISn{n+2}\proves \Rfn{\Pin{}{n+1}}{\CSn n}$. 
	\end{cor}
	\begin{proof}
		Let $\phi(\vec x )$ be a $\Sin{}{n}$ formula. Working in $\ISn{n+2} $, suppose that $\Prov \forall \vec x . \phi$ (so that $\CSn n \proves \forall \vec x . \phi (\vec x)$).
		We may assume that every formula occurring in a $\CSn{n}$ proof of the sequent $\seqar \phi(\vec x)$ is $\Sin{}{n}$, thanks to free-cut elimination, Thm.~\ref{thm:free-cut-elim} (recall that this result is provable already in $\ISn 1$). 
		Thus we have $\Nat, \emptyset \models_{n+1} \forall \vec x . \phi(\vec x)$ by Thm.~\ref{thm:soundness-formalised} and Prop.~\ref{prop:tarski}.\ref{item:tarski-forall}, whence the result follows by the reflection property, Prop.~\ref{prop:reflection}.
	\end{proof}
	

	\noindent
	Notice that, while the statement of Cor.~\ref{cor:reflection} above is peculiar to the current formulation of a $\CSn n$ proof, $\prf_n$, it holds also under any other notion of proof that is provably equivalent in $\ISn{n+2}$.
	In particular, in Sect.~\ref{sect:sect:prf-comp-csn} we discussed another notion of proof for $\CSn{n}$ which, morally, allowed ``free cuts'' to occur inside cycles.
	Since the equivalence of the two formulations, Prop.~\ref{prop:cyclic-sn-general-dfn}, is proved using only free-cut elimination and basic reasoning, all formalisable in $\ISn 1$,
	the version of Cor.~\ref{cor:reflection} for that more liberal notion of a $\CSn n $ proof holds too.
	
	\smallskip
	
	As usual, we may see $\Pin{}{1}$-reflection as another formulation of `consistency'.
	Let us write $\cons{\CSn n }$ for the sentence $\cnot \Prov 0=1$ (notice that this is a $\Pin{}{1}$ sentence).

	\begin{cor}
		[of Thm.~\ref{thm:soundness-formalised}]
		\label{cor:isn+2-prv-cons-csn}
		For $n\geq 0$, we have $\ISn{n+2} \proves \cons{\CSn n}$
	\end{cor}

	\begin{proof}
		Follows immediately from Cor.~\ref{cor:reflection} above by substituting $0=1$ for $\phi$.
	\end{proof}

	\noindent
	We will see in the next subsection that this result is, in fact, optimal with respect to logical complexity.
	
	\subsection{Incompleteness}
	\label{sect:sect:incompleteness-csn}
	Unsurprisingly, all the theories $\CSn n$ suffer from G\"odel's incompleteness theorems.
	Even though $\CSn n$ is not explicitly defined axiomatically, it does have a recursively enumerable notion of provability, namely $\Prov$, and so must be incomplete with respect to this notion (see, e.g., Thm.~2.21 in \cite{HajPud:93}):

	\begin{thm}
		[G\"odel's second incompleteness theorem, for cyclic theories]
		\label{thm:goedel-second-incompleteness}
		For $n\geq 0$, as long as $\CSn n $ is consistent (i.e.\ $\CSn n \notproves 0=1$), we have $\CSn{n} \notproves \cons{\CSn n }$.
	\end{thm}

	\noindent
	Consequently we have that Cor.~\ref{cor:isn+2-prv-cons-csn} is, in fact, optimal in terms of logical complexity:
	
	\begin{cor}
		\label{cor:isn+1-nprv-cons-csn}
		For $n\geq 0 $, we have $\ISn{n+1} \notproves \cons{\CSn n}$.
	\end{cor}
	\begin{proof}
		Suppose otherwise. Then also $\CSn n \proves \cons{\CSn n }$ by $\Pin{}{1}$-conservativity, cf.\ Thm.~\ref{thm:cyclic-sim-ind}, which contradicts G\"odel's second incompleteness above, Thm.~\ref{thm:goedel-second-incompleteness}.
	\end{proof}
	
	\noindent
	We will see in the next section that this has a curious consequence for the reverse mathematics of results in $\omega$-automaton theory.
	
	\section{On the logical strength of McNaughton's theorem}
	\label{sect:red-to-det}
	
	In this section we show how the results of this work yield an unexpected corollary: certain formulations of \emph{McNaugton's theorem}, that every NBA has an equivalent deterministic `parity' or `Muller' automaton, are not provable in $\RCA$. The general question of the logical strength of McNaughton's theorem was notably left open in the recent work \cite{KMPS16:buchi-reverse}.

	Our result is non-uniform in the sense that unprovability holds for any \emph{explicit} primitive recursive determinisation construction. As far as the author is aware this accounts for all known proofs of McNaughton's theorem, suggesting that it is unlikely to be provable at all, in its usual uniform version, in $\RCA$.
	That said, we point out that the statement of McNaughton's theorem itself is arguably not so well-defined in the context of reverse mathematics: it is not clear in $\RCA$ that different versions of the theorem coincide, namely with respect to the choice of (a) acceptance conditions (parity, Muller, etc.) and (b) formulation of the set of states infinitely often hit during a run (negative, $\forall$, vs.\ positive, $\exists$).
	
	Our argument is based on an alternative route to proving the soundness of $\CSn n$. Assuming that an appropriate version of McNaughton's theorem is indeed provable in $\RCA$, we are in fact able to formalise the soundness argument for $\CSn n $ already in $\ISn{ n+1}$. However, consequently we have that $\ISn {n+1}$ proves the consistency of $\CSn n $, and so $\CSn n $ proves its own consistency by $\Pin{}1$-conservativity, cf.~Thm.~\ref{thm:cyclic-sim-ind}, which is absurd by G\"odel's second incompleteness theorem for $\CSn n$, Thm.~\ref{thm:goedel-second-incompleteness}.

	
	%
	
	%
	%
	%
	%
	%
	%
	%
	%
	%

	\subsection{Deterministic parity automata and universality}
	Due to space considerations, we only briefly present the details of parity automata. The reader is encouraged to consult, e.g., \cite{Tho97:aut-chapt}, for further details on automaton theory for $\omega$-languages.

	A (non-deterministic) \textbf{Rabin} or \textbf{parity} automaton (NRA) is a just a NBA where, instead of a set of final states $F$, we have a function $c: Q \to \Nat$, called a \textbf{colouring}.
	A word is accepted by a NRA if it has a run in which the least colour of a state occurring infinitely often is even.
	The notion of deterministic parity automaton (DRA) is analogous to that of a DBA, i.e.\ requiring the transition relation to be deterministic and total.
	
	\begin{thmC}
		[McNaughton, \cite{McN66:determinisation}]
		For every NBA $\mathcal A$, we can effectively construct a DRA accepting the same language.
	\end{thmC}
	
	\noindent
	Actually, McNaughton gave this result for deterministic \emph{Muller} automata rather than parity automata. The equivalence of these two models is well-known though, as we previously mentioned, it is not clear whether $\RCA$ can prove their equivalence.
	The fact that we use parity automata here is arbitrary; we believe a similar exposition could be carried out for Muller automata.
	
	\smallskip

	As for DBA, we may naturally express language acceptance for a DRA $\mathcal A = (A, Q, \delta, q_0, c)$ by an {arithmetical} formula, i.e.\ without SO quantifiers. 
	For our purposes, it will be useful to take a `negative' formulation of acceptance:
	\[
	X \in \lang (\mathcal A)
	\  \dfn \ 
	\forall q \in Q .
	\left(
	\left(
	\begin{array}{rl}
	& \forall x . \exists x' > x .\  q_X (x') = q \\
	\cand & \exists x . \forall x'> x .\  c(q_X (x')) \geq c(q)
	\end{array}
	\right)
	\cimp \text{``$c(q)$ is even''}
	\right)
	\]
	We write $\sigma : q_1 \trarr * \delta q_2$ if a word $\sigma \in A^*$ determines a path along $\delta $ starting at $q_1$ and ending at $q_2$. We write $\trarr + \delta$ when the path is nonempty.
	A \emph{simple loop} about a state $q\in Q$ is a nonempty path along $\delta$ beginning and ending at $q$ that visits no intermediate state more than once.
	
	Recall that we call an $\omega$-automaton \emph{universal} if it accepts all $\omega$-words over its alphabet.
	We write $\univ (\mathcal A)$ for a standard recursive procedure testing {universality} of a DRA $\mathcal A$: ``for every odd-coloured state $q$ reachable from $q_0$, any simple loop about $q$ contains a state coloured by an even number $<c(q)$''. More formally, writing $\sigma' \leq \sigma$ if $\sigma'$ is a prefix of $\sigma$:
	$$
	\univ (\mathcal A)
	\ \dfn \ 
	\begin{array}{l}
	\ \ \forall q \trrra * \delta q_0 . \ \forall \sigma : q \trarr + \delta q. \\
	\left(
	\begin{array}{rl}
	&
	\text{``$\sigma$ is a simple loop''} \cand \text{``$c(q) $ is odd''} \\
	\cimp &
	\exists \sigma' \leq \sigma . \ \exists  q'\in Q . \ (\sigma' : q \trarr + \delta q' \ \cand  \text{``$c(q') $ even''} \cand c(q')<c(q))
	\end{array}
	\right)
	\end{array}
	$$
	Clearly this formula is provably $\Din 0 1 $ in $\RCA$.
	%
	%
	%
	Furthermore:
	\begin{prop}
		\label{prop:rca-prov-corr-univ-dra}
		$\RCA \proves \forall \text{ DRA } \mathcal A . \ (\univ(\mathcal A) \ciff \forall X \in A^\omega . X \in \lang (\mathcal A))$.
	\end{prop}
	
	
	\begin{proof}
		Working in $\RCA$, let $\mathcal A = (A, Q, \delta, q_0, c)$ be a DRA.
		%
		For the left-right implication, suppose there is some $X \in A^\omega$ such that $X \notin \lang (\mathcal A)$. Thus we have some $q \in Q$ such that $c(q)$ is odd and the following hold:
		\begin{equation}
		\label{eqn:inf-occ}
		\forall x . \exists x' > x .\ q_X (x') = q
		\end{equation}
		\begin{equation}
		\label{eqn:event-lb}
		\exists x . \forall x'> x.\ c(q_X (x') ) \geq c(q)
		\end{equation}
		Let $x_0$ be a witness to \eqref{eqn:event-lb}, and let $x_0 < x_1 < x_2$ such that $q_X (x_1) = q_X (x_2) = q$, by two applications of \eqref{eqn:inf-occ}.
		We will need the following intermediate 
		(arithmetical) 
		result,
		\[
		\text{If $\sigma: q \trarr + \delta q$, there is a subsequence $\sigma'$ of $\sigma$ that is a simple loop on $q$.}
		\]
		
		
		\noindent
		which follows directly by induction on $|\sigma|$, eliminating intermediate loops at each inductive step in the case of non-simplicity.
		Now we apply this result to the sequence $(X(x))^{x_2}_{x = x_1}$ to obtain a simple loop about $q$; moreover since this will be a subsequence of $(X(x))^{x_2}_{x = x_1}$, we have that any even-coloured state occurring in it is coloured $> c(q)$, since $x_0$ witnesses \eqref{eqn:event-lb} and $x_0 < x_1 < x_2 $, so $\cnot \univ (\mathcal A)$.
		
		For the right-left implication, we proceed again by contraposition. Suppose $\cnot \univ (\mathcal A)$, and let $\sigma: q_0 \trarr * \delta  q$ and $\tau : q \trarr + \delta q$ such that $c(q)$ is odd, and $\tau$ is a simple loop containing no states coloured $< c(q)$. 
		We may now set $X = \sigma \tau^\omega$ (which is easily defined by comprehension) and show that $X\notin \lang (\mathcal A)$.
		For this it suffices to show \eqref{eqn:inf-occ} and \eqref{eqn:event-lb} above. For the former, given $x$ we set $x ' = |\sigma | + m|\tau| > x$, for some sufficiently large $m$. 
		For the latter, we set $x = |\sigma|$ as the witness to the outer existential, whence \eqref{eqn:event-lb} follows by construction of $\tau$.
		%
		%
		%
		%
		%
	\end{proof}

	
	\subsection{Reducing soundness of $\CDn 0 $ to a version of McNaughton's theorem}
	Henceforth we may write $\lang (\mathcal A) = \lang (\mathcal B)$ as shorthand for $\forall X . (X \in \lang (\mathcal A) \ciff X\in \lang(\mathcal B))$, where $\mathcal A$ and $\mathcal B$ may be any type of automaton thus far encountered with respect to their associated notions of membership.
	Based on our `negative' formulation of DRA acceptance, we define for a (definable) function $d$: $$\McN_d \  \dfn \ 
	\forall \text{ NBA }\mathcal A . 
	(
	\text{``$d(\mathcal A)$ is a DRA''} 
	\cand  
	\lang (\mathcal A) = \lang (d(\mathcal A))
	$$
	
	Assuming this is provable in $\RCA$ for some primitive recursive function $d$, we will reproduce a version of Thm.~\ref{thm:arithmetisation-of-correctness} in $\ISn{n+1}$.
	The idea is that, rather than expressing the fact that $\lang (\mathcal A_1 )\subseteq \lang(\mathcal A_2)$ by saying ``$(\mathcal A_1^c \sqcup \mathcal A_2 )^c $ is empty'', as we did in Sects.~\ref{sect:ind-sim-cyc} and \ref{sect:nonuniform-ind-sim-cyc}, we may rather express it as ``$\mathcal A_1^c \sqcup \mathcal A_2$ is universal'', relying on 
	$\McN_d$ and Prop.~\ref{prop:rca-prov-corr-univ-dra} above.
	
	
	\medskip
	
	For a DBA $\mathcal A_1 $ and a NBA $\mathcal A_2$ we define $\mathcal A_1 \sqsubseteq_d \mathcal A_2$ as
	\(
	\univ (d(\mathcal A_1^c \sqcup \mathcal A_2 )))
	\).
	Let us write $(\ref*{eqn:arith-form-prog-traces}_d)$ for the equation \eqref{eqn:arith-form-prog-traces} with $\sqsubseteq$ replaced by $\sqsubseteq_d$, i.e.:
	\begin{equation}
	\tag{$\ref*{eqn:arith-form-prog-traces}_d$}
	\label{eqn:arith-form-prog-traces'}
	\forall\ \text{DBA}\ \mathcal A_1 , \forall \ \text{NBA}\ \mathcal A_2 .\ 
	\left( 
	(\mathcal A_1 \sqsubseteq_d \mathcal A_2 \cand X \in \lang(\mathcal A_1) )
	\cimp 
	\ArAcc (X, \mathcal A_2 ) 
	\right)
	\end{equation}
	We have the following analogue to Thm.~\ref{thm:arithmetisation-of-correctness}:
	
	\begin{prop}
		\label{prop:rca+mcn-prov-aracc'}
		$\RCA + \McN_d\proves \eqref{eqn:arith-form-prog-traces'}$.
	\end{prop}
	
	\begin{proof}
		Mimicking the proof of Thm.~\ref{thm:arithmetisation-of-correctness},  we work in $\RCA $ and suppose $X\in \lang (\mathcal A_1)$ and $\mathcal A_1 \sqsubseteq_d \mathcal A_2$.
		We have:
		\[
		\begin{array}{rll}
		& \univ (d(\mathcal A_1^c \sqcup \mathcal A_2)) & \text{since $\mathcal A_1 \sqsubseteq_d \mathcal A_2$} \\
		\implies & \forall Y \in A^\omega . Y \in \lang (\mathcal A_1^c \sqcup \mathcal A_2) & \text{by 
			Prop.~\ref{prop:rca-prov-corr-univ-dra} and 
			$\McN_d$} \\
		\implies & \forall Y \in A^\omega. ( Y \in \lang (\mathcal A^c_1) \corr Y \in \lang (\mathcal A_2) ) & \text{by Lemma~\ref{lem:aut-clos-props-in-so-arith}.\ref{item:union-rca} } \\
		\implies &\forall Y \in A^\omega. ( Y \in \lang (\mathcal A_1) \cimp Y \in \lang (\mathcal A_2) )  & \text{by Lemma~\ref{lem:aut-clos-props-in-so-arith}.\ref{item:compl-dba-rca}} \\
		\implies & X \in \lang (\mathcal A_2) & \text{since $X \in \lang(\mathcal A_1)$} \\
		\implies & \ArAcc (X, \mathcal A_2) & \text{by Prop.~\ref{prop:arith-acc}.} \qedhere
		\end{array}
		\]
	\end{proof}
	
	We may use this result to reconstruct the entire formalised soundness argument for $\CSn n$ of Sect.~\ref{sect:ind-sim-cyc} in $\ISn{n+1}$ instead of $\ISn{n+2}$, assuming $\McN_d$. In particular, using Prop.~\ref{prop:rca+mcn-prov-aracc'} above instead of Thm.~\ref{thm:arithmetisation-of-correctness}, we may recover versions of Cor.~\ref{cor:is2-prov-aracc}, Thm.~\ref{thm:soundness-formalised} and Cor.~\ref{cor:reflection} for $\ISn{n+1}$ instead of $\ISn{n+2}$, with respect to \eqref{eqn:arith-form-prog-traces'} instead of \eqref{eqn:arith-form-prog-traces}.
	Formally, let $\Rfnd{\Pin{}{k}}{\CSn n }$ denote the formulation of the $\Pin{}{k}$-reflection principle for $\CSn n $ induced by using $\sqsubseteq_d$ instead of $\sqsubseteq$ throughout Sect.~\ref{sect:sect:rfn-cons}, with respect to the definitions of $\prf_n$ and $\Prov$.
	Similarly, let $\consd{\CSn n }$ be the induced consistency principle.

	\begin{prop}
		\label{prop:if-mcn-then-rfn}
		For $n\geq 0$, if $\RCA \proves \McN_d$ for some primitive recursive function $d$, then $\ISn{n+1} \proves \Rfnd{\Pin{}{n+1}}{\CSn n} $, so
		in particular $\ISn{n+1} \proves \consd{\CSn n }$.
	\end{prop}
	\begin{proof}
		[Proof sketch]
		The argument goes through just like that of Cors.~\ref{cor:reflection} and \ref{cor:isn+2-prv-cons-csn} of Thm.~\ref{thm:soundness-formalised}, except that we use Prop.~\ref{prop:rca+mcn-prov-aracc'} instead of Thm.~\ref{thm:arithmetisation-of-correctness}. Appealing to the assumption that $\RCA \proves \McN_d$, this version of the argument requires only $\CIND{\Sin{0}{1}}$ instead of $\CIND{\Sin{0}{2}}$, thus yielding $\ISn{n+1}$ proofs overall once we substitute the appropriate formulae for $X$.
	\end{proof}
	
	\begin{thm}
		$\RCA \notproves \McN_d$, for any primitive recursive function $d$.
	\end{thm}
	\begin{proof}
		[Proof sketch]
		The same argument as Cor.~\ref{cor:isn+1-nprv-cons-csn} holds for our revised notion of consistency; in particular we have that $\ISn 1 \notproves \consd{\CDn 0 }$.
		The result now follows immediately by the contraposition of Prop.~\ref{prop:if-mcn-then-rfn} above, for $n=0$.
	\end{proof}

	\section{Conclusions and further remarks}
	\label{sect:conc}
	In this work we developed the theory of cyclic arithmetic by studying the logical complexity of its proofs.
	We showed that inductive and cyclic proofs of the same theorems require similar logical complexity, and obtained tight quantifier complexity bounds in both directions.
	We further showed that the proof complexity of the two frameworks differs only elementarily, although it remains unclear how to properly measure proof complexity for the fragments $\CSn n$, even if the theory itself seems well-defined and robust.
	Many of these issues constitute avenues for further work.

	\subsection{Comparison to the proofs of \cite{BerTat17:lics} and \cite{Sim17:cyclic-arith}}
	One reason for our improved quantifier complexity compared to \cite{Sim17:cyclic-arith}, is that Simpson rather relies on \emph{Weak K\"onig's Lemma} ($\wkl$) to obtain an infinite branch when formalising the soundness argument for cyclic proofs.
	This, \emph{a priori}, increases quantifier complexity of the argument, since $\wkl$ is known to be unprovable in $\RCA$ even in the presence of $\CIND{\Sin{0}{2}}$;
	in fact, it is incomparable to $\CIND{\Sin{0}{2}}$ (see, e.g., \cite{KMPS16:buchi-reverse}).
	That said, we believe that the `bounded-width $\wkl$' ($\bwwkl$) of \cite{KMPS16:buchi-reverse} should suffice to carry out Simpson's proof, and this principle is provable already in $\RCA + \CIND{\Sin 0 2}$.
	Applying this strategy to his proof yields only that $\CSn n \subseteq \ISn{n+3}$, since $\bwwkl$ is applied to a $\Pin{0}{n+1}$ set, though this should improve to a $\ISn{n+2}$ bound by using the non-uniform version of NBA complementation implicit in \cite{KMPS16:buchi-reverse}, cf.~Prop.~\ref{prop:nonuniform-compl}.
	We reiterate that the main improvement here is in giving a uniform formulation of those results; not only does this lead to a better proof complexity result, cf.~Thm.~\ref{thm:elementary-simulation}, but we also recover a metamathematical account of the theories $\CSn n$, cf.~\ref{sect:further-metalogical}.
	%
	%
	%
	%

	Berardi and Tatsuta's approach, \cite{BerTat17:lics}, is rather interesting since it is arguably more `structural' in nature, relying on proof-level manipulations rather than reflection principles.
	That said there are still crucial sources of logical complexity blowup, namely in an `arithmetical' version of \emph{Ramsey's theorem} (Thm.~5.2) and the consequent \emph{Podelski-Rybalchenko termination theorem} (Thm.~6.1).
	Both of these apparently increase quantifier complexity by several levels, and so their approach does not seem to yield comparable logical bounds to this work.
	Since proof complexity is not a primary consideration of their work, it is not simple to track the precise bounds in \cite{BerTat17:lics}. There are some apparent sources of exponential blowups,\footnote{For instance, Lemma 8.4 in that work yields a set of apparently exponential size in the worst case, and this bounds from below the size of the overall translation, e.g.\ as in Lemma 8.7.} 		
	though it seems that the global simulation is elementary.
	As before, we reiterate that the major improvement in the present work is in the uniformity of our exposition: the approach of \cite{BerTat17:lics} is fundamentally non-uniform so does not yield any metamathematical account of cyclic arithmetic.
	

	\subsection{On the correctness criteria for cyclic proofs}
	Since the algorithms used to check correctness of a cyclic preproof reduce to the inclusion of B\"uchi automata, the exponential simulation of $\CA$ by $\PA$ is optimal, unless there is a nondeterministic subexponential-time algorithm for $\pspace$ or, more interestingly, there is an easier way to check cyclic proof correctness. (In fact, technically, it would suffice to have an easier criterion for a \emph{larger} class of preproofs that were, nonetheless, sound.)
	As far as we know, $\pspace$ remains the best known upper bound for checking the correctness of general cyclic preproofs, although efficient algorithms have recently been proposed for less general correctness criteria, cf.~\cite{Stratulat17,Nol-Sau-Tas:18:local-validity}. Thus, it would be interesting to prove a corresponding lower bound or otherwise improve the upper bound.
	Conditional such results could be obtained via, say, certain polynomial upper bounds on proof complexity in $\CA$: for instance, if $\CA$ were to have polynomial-size proofs of each correct B\"uchi inclusion then cyclic proof correctness would not be polynomial-time checkable, unless $\np =\pspace$.
	Unfortunately na\"ive attempts at this approach fail, but the general question of whether $\PA$ and $\CA$ are exponentially separated seems pertinent.
	
	On the other hand, the translation of Lemma~\ref{lem:ind-to-cyc-trans} from inductive proofs to cyclic proofs is rather structured. In light of the converse result in Sect.~\ref{sect:ind-sim-cyc} it might make sense in further work, from the point of view of logical complexity, to consider only cyclic proofs accepted by some weaker more efficiently verified criterion, such as \cite{Stratulat17,Nol-Sau-Tas:18:local-validity}.
	
	\subsection{Interpreting ordinary inductive definitions in arithmetic}
	\label{sect:sect:interp-folid-arith}
	
	
	%
	%
	%
	%
	%
	%

	In earlier work by Brotherston and Simpson, cyclic proofs were rather considered over a system of FO logic extended by `ordinary' \emph{Martin-L\"of} inductive definitions \cite{Pml71:haupt-int-iid}, known as $\FOLID$ \cite{Bro06:phdthesis,BroSim07:comp-seq-calc-ind-inf-desc,BroSim11:seq-calc-ind-inf-desc}.
	Berardi and Tatsuta showed in \cite{BerTat17:lics} that the cyclic system $\CLKID^\omega$ for $\FOLID$ is equivalent to the inductive system $\LKID$, when at least arithmetic is present, somewhat generalising Simpson's result \cite{Sim17:cyclic-arith}.
	We point out that 
	ordinary Martin-L\"of inductive definitions can be \emph{interpreted} in arithmetic
	in the usual way by a $\Sin{}{1}$ inductive construction of `approximants', and a proof of $\CLKID^\omega$ may be similarly interpreted line-by-line in $\CA$.
	{(This is similar to the role of the `stage number predicates' in \cite{BerTat17:lics}.)}
	In particular, this means that $\CLKID^\omega (+ \PA)$ is \emph{conservative} over $\CA$.
	We reiterate that the interest behind the results of \cite{BerTat17:lics} is rather the structural nature of the transformations, but this observation also exemplifies why $\CA$ is a natural and canonical object of study, as argued in \cite{Sim17:cyclic-arith}.
	
	
	
	%

	\subsection{Cyclic propositional proof complexity}
	
	
One perspective gained from this work comes in the setting of \emph{propositional proof complexity} (see, e.g., \cite{Cook:2010:LFP:1734064,Krajicek:1996:BAP:225488}).
Thm.~\ref{thm:cyclic-sim-ind} of Sect.~\ref{sect:cyc-sim-ind} should relativise to theories with oracles too. For instance, we may formalise in $\CDn 0 (f)$, where $f$ is a fresh (uninterpreted) function symbol, a proof of the relativised version of the (finitary) pigeonhole principle (see App.~\ref{sect:php-case-study}).
This formula is known to be unprovable in $\IDn 0 (f)$ due to lower bounds on propositional proofs of bounded depth \cite{KraPudWood:95:ExpPHPbdF,PitBeaImp:93:ExpLBPHP}.
	
	At the same time the `Paris-Wilkie' translation \cite{paris19810}, which fundamentally links $\IDn{0}(f)$ to bounded-depth propositional proofs, works locally on an arithmetic proof, at the level of formulae.
	Consequently one may still apply the translation to the lines of a $\CDn{ 0}(f)$ proof to obtain small `proof-like' objects containing only formulae of bounded depth, and a cyclic proof structure.
	One would expect that this corresponds to some strong form of `extension', since it is known that adding usual extension to bounded systems already yields full `extended Frege' proofs.
	However at the same time, some of this power has been devolved to the proof structure rather than simply at the level of the formula, and so could yield insights into how to prove simulations between fragments of Hilbert-Frege systems with extension.
	
	We point out that recent work, \cite{AtsLaur:18}, relating cyclic proof structures to proof complexity has already appeared, albeit with a different correctness criterion.

	\section*{Acknowledgments}
	\noindent 
	I am indebted to Alex Simpson for encouraging me to pursue this work and for his valuable feedback. Similarly, I would like to thank Stefano Berardi for several illuminating discussions on metalogical matters regarding cyclic proofs.
	Finally, I would like to thank James Brotherston, Guilhem Jaber, Alexis Saurin and the anonymous reviewers for this and previous versions of this work for all their helpful comments and insights.

	\bibliographystyle{alpha}
	\bibliography{cyclicarithmetic-refs}
	
	\newpage
	\appendix
	
	\section{Case study: the relativised pigeonhole principle}
	\label{sect:php-case-study}
	In this section we will give an example of the translation from Sect.~\ref{sect:cyc-sim-ind} in a relativised setting.
	Simpson already gave an example of a separation between $\CSn 1$ and $\ISn{1}$ via the totality of the Ackermann-P\'eter function \cite{Sim17:cyclic-arith}, a $\Pin {} 2$ sentence. Logically simpler $\Pin{}{1}$ separations are obtainable in the form of consistency principles, as we discussed in Sect.~\ref{sect:further-metalogical}.
	
	In this section we consider the well-known \emph{pigeonhole principle},
	%
	defined by the following FO formula with an uninterpreted function symbol $f$:
	\[
	\PHP(f)
	\ \dfn \ 
	\forall n . ( \forall x \leq n . f(x)<n 
	\ \cimp\ 
	\exists x \leq n .\exists x'<x . f(x)=f(x')
	)
	\]
	
	\noindent
	It is well-known that $\IDn 0 (f)$ does not prove $\PHP(f)$, due to lower bounds on propositional proofs of bounded depth \cite{KraPudWood:95:ExpPHPbdF,PitBeaImp:93:ExpLBPHP}. On the other hand, by relativising the constructions in Sect.~\ref{sect:cyc-sim-ind}, it is provable in $\CDn 0 (f)$ thanks to the known simple proofs in $\ISn 1 (f)$.
	
	\subsection{A simple proof of $\PHP (f)$ in $\ISn 1 (f)$}
	First we recall a simple well-known proof of $\PHP(f)$ in $\ISn{1} (X)$.
	In fact, as in Sect.~\ref{sect:cyc-sim-ind}, we will work with $\CIND{\Pin{}{1}}$ rather than $\CIND{\Sin{}{1}}$.

	Temporarily, let us write $A,B$ for first-order variables that we interpret as the codes of finite sets.
	For such codes we may use set-theoretic symbols such as $\inn$ and $\setminus$ with their usual interpretations, with the understanding that their basic properties are provable in $\IDn 0 $.
	\begin{lem}
		\label{lem:php-ind}
		$\ISn 1 (f)$ proves the following:
		\begin{equation}
		\label{eqn:php-ih}
		\forall A , B . (|A|>|B| \cimp 
		(\forall x \inn A . f(x) \inn B \cimp \exists x ,x' \inn A . (x \neq x' \cand f(x)=f(x')))
		)\end{equation}
	\end{lem}
	\begin{proof}
		Working in $\ISn 1 (f)$, we reason by induction on $|B|$.
		If $B$ is empty and $|A|>|B|$ then $A$ is nonempty and so \eqref{eqn:php-ih} is vacuously true by falsity of the premiss.
		
		Otherwise $B$ is nonempty, so let $b \in B$ and let $|A|>|B|$.
		\begin{itemize}
			\item If $\exists x \in A . f(x) =b$ then, let $a\in A$ such that $f(a)=b$.
			\begin{itemize}
				\item If $\exists x\in A . (x'\neq a \cand f(x')=b)$ then we are done.
				\item Otherwise suppose $\forall x \in A. (f(x)=b\cimp x=a)$.
				Then we have $\forall x \in A\setminus \{a\} . f(x) \in B \setminus \{b\} $. Since we still have that $|A\setminus\{a\}|>|B\setminus\{b \}$ we may conclude by the inductive hypothesis.
			\end{itemize}
			\item Otherwise $\forall x \inn A . f(x) \neq b$, so in fact $\forall x \in A . f(x) \inn B \setminus\{b\}$ and still $|A|>|B\setminus\{b\}$. Hence we conclude by the inductive hypothesis.\qedhere
		\end{itemize}
	\end{proof}
	
	From here there is a simple proof of $\eqref{eqn:php-ih} \cimp \PHP(f)$ in $\IDn 0 (f)$, by instantiating $A$ and $B$ as $[0,n]$ and $[0,n)$ resp. Thus we have that $\ISn 1 (f) \proves \PHP(f)$.

	\subsection{A proof of $\PHP(f)$ in $\CDn 0 (f)$}
	To show that $\CDn 0 (f) \proves \PHP(f)$ it suffices to give $\CDn 0 (f)$ proofs of Lemma~\ref{lem:php-ind}. The remainder of the argument may be carried out in $\IDn 0 (f)$, and so also in $\CDn 0 (f)$ by Prop.~\ref{prop:isn-in-csn}, 
	
	\begin{lem}
		$\CDn 0 (f) \proves \eqref{eqn:php-ih}$.
	\end{lem}
	\begin{proof}
		As abbreviations, let us write $f(A)\subseteq B$ for $\forall x \in A.f(x) \in B$ and $\Inj (A)$ for $\forall x , x' \in A.( f(x)= f(x') \cimp x  = x')$.
We give an appropriate derivation in Fig.~\ref{fig:php}, mimicking the argument of Lemma~\ref{lem:php-ind} under Lemma~\ref{lem:ind-to-cyc-trans}, where $\pi_0$ is a $\IDn 0 (f)$ proof of $B= \emptyset , |A| > |B| , f(A)\subseteq B  \seqar \cnot \Inj (A)$ and $\pi_1$ is a proof of $a'\in A , a'\neq a , f(a')=b, a \in A, f(a)=b, b\in B, |A| > |B| , f(A)\subseteq B  \seqar \cnot \Inj (A) $. 
	\end{proof}

\begin{landscape}
\begin{figure}
	\vspace{.2\textheight}
	\[
		\small
	\vlderivation{
		\vlin{\cimp, \forall}{}{\seqar \forall A , B . (|A|>|B| \cimp 
			(f(A) \subseteq B \cimp \cnot \Inj (A))) }{
			\vliin{}{\bullet}{|A| > |B| , f(A)\subseteq B  \seqar \cnot \Inj (A) }{\vlhy{\pi_0}}{
				\vliin{}{}{b\in B, |A| > |B| , f(A)\subseteq B  \seqar \cnot \Inj (A) }{
					\vliin{}{}{a \in A , f(a)=b, b\in B, |A| > |B| , f(A)\subseteq B  \seqar \cnot \Inj (A)}{
						\vlin{}{}{\forall x \in A . (f(x) = b \cimp x = a), a \in A, f(a)=b, |A| > |B| , f(A)\subseteq B  \seqar \cnot \Inj (A)}{
							\vlin{\sub}{}{|A\setminus\{a\}| > |B\setminus\{b\}|, f(A\setminus \{a\}) \subseteq B\setminus\{b\}\seqar \cnot \Inj (A\setminus\{a\})}{
								\vlin{}{\bullet}{|A| > |B| , f(A)\subseteq B  \seqar \cnot \Inj (A) }{\vlhy{\vdots}}
							}
						}
					}{\vlhy{\pi_1}	}
				}{
					\vlin{}{}{\forall x \in A . \cnot f(x) = b , b \in B, |A| > |B| , f(A)\subseteq B  \seqar \cnot \Inj (A) }{
						\vlin{\sub}{}{ |A| > |B\setminus \{b\}| , f(A)\subseteq B\setminus \{b\}  \seqar \cnot \Inj (A) }{
							\vlin{}{\bullet}{|A| > |B| , f(A)\subseteq B  \seqar \cnot \Inj (A) }{\vlhy{\vdots}}
						}
					}
				}
			}
		}
	}
	\]
	\caption{A $\CDn 0 (f)$ proof of $\PHP(f)$, where $\pi_0$ is a $\IDn 0 (f)$ proof of $B= \emptyset , |A| > |B| , f(A)\subseteq B  \seqar \cnot \Inj (A)$ and $\pi_1$ is a proof of $a'\in A , a'\neq a , f(a')=b, a \in A, f(a)=b, |A| > |B| , f(A)\subseteq B  \seqar \cnot \Inj (A) $.}
	\label{fig:php}
\end{figure}
\end{landscape}

\section{Non-uniform complementation of B\"uchi automata}
\label{sect:complementation-nonuniform}
In this section we give a self-contained proof of Prop.~\ref{prop:nonuniform-compl}, which is only implicit in \cite{KMPS16:buchi-reverse}. 
One novel contribution here is a much simpler proof of the (Nonuniform) Additive Ramsey Theorem.
The remainder of the complementation argument is standard and follows closely \cite{KMPS16:buchi-reverse}, though we present it here with more structure and proof details.

\subsection{Nonuniform Additive Ramsey Theorem in $\RCA$}
Let us write $\binom \Nat 2$ for the set of unordered pairs of natural numbers.
We in fact write its elements as ordered pairs $(i,j)$ where always $i<j$.

Let $(S, \bullet)$ be a finite semigroup and consider a `colouring' $C : \binom \Nat 2 \to S$.
We may omit the group operation symbol $\bullet$ when composing elements of $S$.
We say that $C$ is \emph{additive} if, whenever $i<j<k$, we have $C(i,j)C(j,k) = C(i,k)$.
We say that $I\subseteq \Nat$ is an \emph{$a$-clique} (under $C$) if, $\forall i,j\in I$ we have $C(i,j) =a $.

\begin{thm}
	[Nonuniform Additive Ramsey Theorem]
	\label{thm:additive-ramsey-nonunif}
	Let $(S,\bullet)$ be a finite semigroup. Then:
	\[
	\RCA \proves \forall C : \binom \Nat 2 \to S . (\text{``$C$ is additive''} \cimp \exists a \in S. \exists I \subseteq \Nat . 
	\text{``$I$ is an infinite $a$-clique''}   
	)
	\]
	\end{thm}

\noindent
Before we give the proof, we better state the following fact:
\begin{fact}
	[Nonuniform Infinite Pigeonhole Principle]
	\label{fact:iphp-nonunif}
	Let $S$ be a finite set. Then:
$$
\RCA \proves \forall f: \Nat \to S. \exists a \in S . \forall m. \exists x > m . f(x) = a
$$
In particular:
\[
\RCA\proves \forall f : \Nat \to S .\  \exists Y \text{infinite}.\ \forall x,y \in Y .\  f(x) = f(y)
\]
\end{fact}

\noindent
This result is well-known in reverse mathematics and can be proved by a routine meta-level induction on the size of $S$.
The second display follows from the first part by $\CCA{\Din 0 1}$.


We can now give a proof of the Nonuniform Additive Ramsey Theorem.
This argument differs from and is somewhat simpler than the analogous argument from \cite{KMPS16:buchi-reverse} (Prop.~4.1), which requires a detour via the Ordered Ramsey Theorem and `Green' theory. 
In particular, since we only care about nonuniform provability here, we are not conerned about the quantifier complexity of the inductive invariant, since this induction will take place at the meta-level.

\begin{proof}
	[Proof of Thm.~\ref{thm:additive-ramsey-nonunif}]
	We proceed by a meta-level induction on the size of $S$.
	If $S$ is empty the statement is vacuously true, so we proceed to the inductive step.
	
	If for some $a \in S$ there are only finitely many $i$ such that $C(i,j) = a$ for some $j>i$. 
	Then, letting $n$ be such that $C(i,j) \neq a $ for $i,j\geq n$, we may simply apply the inductive hypothesis to the colouring $C(\cdot + n, \cdot + n)$.
	Thus we may assume henceforth that:
	\begin{equation}
	\label{eqn:each-colour-occurs-inf-high}
		\forall a \in S . \exists^\infty i \in \Nat. \exists j>i . C(i,j) = a
	\end{equation}

	Now, suppose that for some $a \in S$ we have $aS \subsetneq S$.
	We will define, assuming the above display, a certain `increasing' enumeration of elements of $\binom \Nat 2$ that map to $a$.
%
%
	Consider the functions $i:\Nat \to \Nat$ and $j: \Nat \to \Nat$ defined simultaneously as follows:
	\[
	\begin{array}{rcl}
	i(0) & \dfn & \beta_0 \mu k. (k = \pair{i'}{j'}  \cand j'> i' \cand C(i', j') = a) \\
	i(n+1) & \dfn & \beta_0 \mu k. (k=\pair{i'}{j'} \cand j'> i'>j(n) \cand C(i', j') = a) \\
	\noalign{\medskip}
	j(0) & \dfn & \mu j'>i(0) . C(i(0),j' ) = a \\
	j(n+1) & \dfn & \mu j'>i(n+1) . C(i(n+1), j') = a
	\end{array}
	\]
	where we write $\mu x . \phi(x)$ for the least $x$ such that $\phi(x)$. Notice that this definition is a form of simultaneous primitive recursion with provably terminating blind searches, assuming Eqn.~\eqref{eqn:each-colour-occurs-inf-high}, so $i $ and $j$ are provably recursive and their defining equations above are provable (in $\RCA + \eqref{eqn:each-colour-occurs-inf-high}$).
	Thus we have that for all $n \in \Nat$:
	\[
	\begin{array}{l}
	i(n) < j(n) < i(n+1) \\
	C(i(n),j(n)) = a
	\end{array}
	\]
	But now, for $m<n$, notice that, 
	$$
	\begin{array}{rcl}
	C(i(m),i(n))& = & C(i(m),j(m))  C(j(m),i(n)) \\
	& = & aC(j(m), i(n)) \\
	& \in & aS \subsetneq S
	\end{array} 
	$$ 
	by the additivity property and assumption. 
	So from here we may apply the inductive hypothesis to the colouring $C(i(\cdot),i(\cdot))$ and conclude.
%
%
%
%
%
%
	Thus we may henceforth assume that:
	\begin{equation}
	\label{eqn:latin-square}
	\forall a \in S. aS = S
	\end{equation}
	In other words, $a \bullet \cdot$ is a bijection on $S$, for any $a \in S$.
	Thus we have the left cancellation property: if $ab = ac $ then $b=c$, and for any $a$ there is a unique $b \in S$ such that $ab=a$.
	
	Now, $C(0, \cdot)$ must take some value $a \in S$ infinitely often, by the non-uniform infinite pigeonhole principle, Fact~\ref{fact:iphp-nonunif}.
	In this case we may simply set $I  = \{ i : C(0,i)=a \}  $.
	Notice that, for $i,j \in I$ with $i<j$, we have $ a= C(0,j) = C(0,i)C(i,j) = aC(i,j)$, and so indeed each $C(i,j)$ is identical, by left-cancellation.
\end{proof}

%

\subsection{Characterising rejection via the Ramseyan factorisation of $\omega$-words}
For the remainder of this section let us fix an NBA $\mathcal A = (A, Q, \delta, q_0, F)$.

Since we deal with non-deterministic automata, it no longer makes sense to use the notation $\sigma : q \trarr{*}{\delta} q'$ from Sect.~\ref{sect:red-to-det}, since there may be several paths through $\delta$ from $q$ to $q'$ following $\sigma$.
Instead, for $\sigma \in A^*$, we write:
\[
x : q\trarr{\sigma}{\delta}q'
\quad \dfn \quad
\text{``$x$ is a path through $\delta$ from $q$ to $q'$ following $\sigma$''}
\]

%
%

Now, for each finite word $\sigma \in A^*$ we may consider its \emph{transition matrix} $\delta(\sigma)$, which is the graph whose nodes are those of $Q$ with an edge from $q $ to $q' $ if there exists some $x:q \trarr \sigma \delta q'$.
The edge is labelled with $\final$ if there is such a path that hits an accepting state (after $q$).
We construe such graphs as functions of type $Q \times Q \to \{0,1,\final \}$ in the natural way.
Let us call the set of all such graphs $\delta(A^*)$, which we note must be finite.
Notice that $\delta: A^* \to (Q \times Q \to \{0,1,\final \})$ is provably recursive in $\RCA$ by the following definition:
\[
\delta(\sigma) (q, q') = 
\begin{cases}
	0 & \nexists x : q \trarr \sigma \delta q' 
\\
\infty & \exists x: q \trarr \sigma\delta q' . \exists i<|x| . x(i) \in F
\\
1 & \text{otherwise}
\end{cases}
\]

We may compose such graphs by a variation of the usual relational composition accounting for the labelling in a natural way: for $\beta, \gamma : Q \times Q \to \{ 0,1,\final \}$, we define $\beta\gamma : Q \times Q \to \{ 0,1,\final \}$ by:
$$(\beta\gamma)(q,q' ) = \max\limits_{p\in Q} \left(\beta(q,p)\cdot \gamma(p,q')\right)$$
Here we define $\max$ and $\cdot$ as expected, in particular setting $0\cdot \final = 0 = \final \cdot 0$.
Intuitively, this is just the same as relational composition, only recording if it is possible to hit a final state en route.
Under these operations
we have that $\delta$ is in fact a homomorphism $A^* \to (Q \times Q \to \{0,1,\final \})$, provably in $\RCA$:

\begin{prop}
\label{prop:trans-matrix-additive}
$\RCA  \proves \forall \sigma, \tau \in A^*  . \delta(\sigma\tau) = \delta(\sigma)\delta(\tau)$.
\end{prop}

\noindent
The proof follows by a straightforward induction on the length of $\tau$. Only the base case when $\tau$ is some $a\in A$ is interesting, with the inductive case following by associativity of word composition.

%
%
%

Now, for $\beta,\gamma: Q\times Q \to \{ 0,1,\final \}$, let us say that the pair $(\beta,\gamma)$ is \emph{rejecting} if:
\begin{enumerate}
\item\label{item:ram-fac-reaches} $\beta\gamma = \beta$; and,
\item\label{item:ram-fac-loops} $\gamma\gamma = \gamma$; and,
\item\label{item:ram-fac-nonfinal} $\forall q \in Q. (\beta(q_0,q) >0 \cimp \gamma(q,q) < \final)$
\end{enumerate}

Again, the property of being rejecting pair is clearly $\Delta^0_1$, since we have fixed $Q$ in advance so there are only finitely many cases to consider.
Let us adopt the interval notation $[i,j]$ for the set $\{i, i+1, \dots, j \}$ and $[i,j)$ for the set $\{ i, i+1, \dots, j-1  \}$.
For an infinite sequence $X$ we also write $X[i,j]$ and $X[i,j)$ for the finite subsequences $(X(i), X(i+1), \dots, X(j))$ and $(X(i), X(i+1), \dots , X(j-1))$ respectively.

\begin{lem}
	[Ramseyan factorisation (in $\RCA$)]
	\label{lem:decomposition-of-word}
	For any $X \in A^\omega$, there are $\beta, \gamma \in \delta(A^*)$ and
	an infinite set $I \subseteq \Nat$ such that:
	\begin{enumerate}
\item\label{item:exists-beta}  $\delta( X[0,i) ) = \beta$, for $i\in I$.
\item\label{item:exists-gamma}  $\delta(X[i , j) ) = \gamma$, for $i,j \in I$ with $i<j$.
\end{enumerate}
Moreover, for any such $\beta,\gamma,I$ satisfiying \eqref{item:exists-beta} and \eqref{item:exists-gamma} above, we have that:
\begin{enumerate}
			\setcounter{enumi}{2}
	\item\label{item:rejecting-pair} $(\beta,\gamma)$ is a rejecting pair if and only if $X \notin \lang (\mathcal A)$. 
	\end{enumerate}
\end{lem}
\begin{proof}
	Working inside $\RCA$, let $C:\binom \Nat 2 \to \delta(A^*)$ by $C(i,j) = \delta(X[i,j))$.
	By Prop.~\ref{prop:trans-matrix-additive} we have that $C$ is an additive colouring, and thus we may apply the Nonuniform Additive Ramsey Theorem, Thm.~\ref{thm:additive-ramsey-nonunif}, to obtain some infinite set $I_0$ with $\gamma = C(i,j) $, for all $i,j \in I_0$ with $i<j$, yielding \eqref{item:exists-gamma}.
	Now, for an arbitrary $i_0\in I_0$ we may set $I = \{ i \in I_0 : i>i_0 \}$.
	Notice that, now, $C(0,i) = C(0,i_0)C(i_0,i) = C(0,i_0)\gamma$, so we may set $\beta = C(0,i)$ for some/any $i \in I$, yielding \eqref{item:exists-beta}.
%
%

For \eqref{item:rejecting-pair},
first suppose $(\beta,\gamma)$ is a rejecting pair and let
$Y \in Q^\omega$ 
be a run of $\mathcal A $ on $X$. 
We will show that $Y$ cannot be accepting.
By the Non-Uniform Infinite Pigeonhole Principle, Fact~\ref{fact:iphp-nonunif}, 
there is some $q \in Q $ and some infinite subset $I'\subseteq I$ such that $Y(i) = q $ for each $i \in I'$.
We claim that, for any $i,j \in I'$ and $k\in \Nat$ such that $i<k<j$, we have that $q_k \notin F$.
For this notice that:
\begin{itemize}
	\item $\beta(q_0,q)>0$ since $Y[0,i]: q_0 \longtrarr{X[0,i)}\delta q$.
	\item  $\gamma (q,q) <\final$ since $(\beta,\gamma)$ is a rejecting pair.
\end{itemize}
So, if $q_k$ lies on the path $Y[i,j]:q \longtrarr{X[i,j)}{\delta}q$, we must have that $q_k \notin F$, appealing to Prop.~\ref{prop:trans-matrix-additive}.
Thus, for any $i\in I'$, we have that $\forall k> i . Y(k)\notin F$, and so $Y$ is not an accepting run.


Now, suppose that $(\beta,\gamma)$ is not a rejecting pair and, by definition, let $q$ be such that $\beta(q_0,q)>0$ and $\gamma(q,q) = \final$.
We may enumerate infinite sets in $\RCA$ so let $I = (i_j)_{j\in \Nat}$.
We may now simply define an accepting run $Y$ of $\mathcal A$ on $X$ by recursive comprehension by insisting that,
\begin{itemize}
	\item in the interval $[0,i_0)$, $Y$ follows the `least' path through $\delta$ from $q_0$ to $q$; and,
	\item in the interval $[i_j, i_{j+1})$, $Y$ follows the `least' path through $\delta$ from $q$ to $q$ hitting a final state.
\end{itemize} 
Such paths must exist since $(\beta,\gamma)$ is not a rejecting pair, and the set of all paths of bounded length may be enumerated in $\RCA$.
By construction, $Y$ is accepting.
\end{proof}

\subsection{The complement NBA and proof of correctness}
Now we are ready to define the complement automaton of $\mathcal A$, which simply guesses a rejecting Ramseyan factorisation of an input $\omega$-word:
\begin{defi}
	
	We define the NBA ${\mathcal A^c} = (A, Q^c, \delta^c, q_0^c, F^c )$ as follows:
	\begin{itemize}
		\item $ Q^c = \{q_0 \} \cup \delta(A^*) \cup \delta(A^*)^2 \cup \delta(A^*)^3$.
		\item $\delta^c$ consists of the following transitions:
		\begin{itemize}
			\item $(q_0, a , (\beta,\gamma, \delta(a))$, for each rejecting pair $(\beta,\gamma)$.
			\item $((\beta,\gamma, \zeta), a , ( \beta, \gamma, \zeta\delta(a) ) )$
			\item $((\beta,\gamma,\zeta) ,a, \gamma )$ if $\zeta\delta(a)=\beta$.
			\item $(\gamma, a, (\gamma, \delta(a)) )$.
			\item $((\gamma, \zeta),a,(\gamma, \zeta\delta(a)))$.
			\item $((\gamma, \zeta),a,\gamma)$ if $\zeta \delta(a)=\gamma$.
			\item $(\gamma,a,\gamma)$ if $\delta(a)=\gamma$.
		\end{itemize}
		\item $q_0^c = q_0$.
		\item $F^c = \delta(A^*)$.
	\end{itemize}
	\end{defi}
\noindent
Now we are ready to prove the non-uniform complementation result.
\begin{proof}
	[Proof of Prop.~\ref{prop:nonuniform-compl}]
By Lemma~\ref{lem:decomposition-of-word}, it suffices to show that $\mathcal A^c$ accepts an $\omega$-word just if has a Ramseyan factorisation that is rejecting.
Let us proceed working inside $\RCA$.

First, suppose $X \in A^\omega$ with $X \notin \lang (\mathcal A)$ and
let $\beta, \gamma, I$ be obtained from Lemma~\ref{lem:decomposition-of-word}. 
Again we may enumerate $I = (i_j)_{j \in \Nat}$.
Define the run $Y$ of $X$ on $\mathcal A^c$ as follows:
\begin{itemize}
	\item $Y(0) = q_0$ and $Y(1) = (\beta, \gamma, \delta(X(0)))$;
	\item in the interval $[1,i_0)$, $Y$ follows the (unique) transitions of the form $((\beta, \gamma, \zeta), a, (\beta, \gamma, \zeta \delta(a) ))$;
		\item $Y(i_j) = \gamma$, for all $j\in \Nat$;
	\item in an interval $(i_{j}, i_{j+1})$, $Y$ follows the (unique) transitions of the form $( (\gamma, \zeta) , a, (\gamma, \zeta\delta(a) ))$.
\end{itemize}
$Y$ is clearly recursive and so is indeed definable by $\CCA{\Din 0 1 }$.
Moreover, $Y$ hits the final state $\gamma$ infinitely often (at each $i_j \in I$), so $Y$ is an accepting run for $X$ on $\mathcal A^c$.

Conversely, suppose $X \in \lang(\mathcal A^c)$ and let $Y$ be an accepting run with $Y(1)=(\beta,\gamma,\delta(X(0)))$.
By a routine induction, $Y$ must hit $\gamma$ infinitely often, since that is the only state in $\delta(A^*)$ that can ever be hit.
Thus, by $\CCA{\Din 0 1}$, we have an infinite set $I = \{ i \in \Nat : Y(i) = \gamma\}$, which we again enumerate $I = (i_j)_{j\in \Nat}$.
We may now show the properties \eqref{item:ram-fac-reaches} and \eqref{item:ram-fac-loops} of Lemma~\ref{lem:decomposition-of-word} with respect to the $\beta, \gamma, I$ thus defined by induction.
More precisely, we prove by induction on $i\in \Nat$ that $Y(i)$ has the following form:
\begin{itemize}
	\item $q_0$, if $i=0$;
	\item $(\beta,\gamma,\delta(X[0,i)))$, if $i<i_0$;
	\item $\gamma$, if $i \in I$; 
	\item $(\gamma, \delta(X[i_j,i)))$ if $i_j < i < i_{j+1}$.
\end{itemize}

\noindent
Properties
\eqref{item:ram-fac-reaches} and \eqref{item:ram-fac-loops} from Lemma~\ref{lem:decomposition-of-word} now follow as special cases. 
Since $(\beta,\gamma)$ was a rejecting pair, by definition of $\mathcal A^c$, we have from Lemma~\ref{lem:decomposition-of-word}.\eqref{item:ram-fac-nonfinal} that $\mathcal A$ rejects $X$.
This concludes the proof.
\end{proof}

\end{document}